\documentclass[a4paper,english,10pt]{amsart} 
\usepackage{amsmath}
\usepackage{amsthm}
\usepackage{amsfonts}
\usepackage{amssymb}
\usepackage{graphicx}
\usepackage{xcolor}
\usepackage{enumerate}
\usepackage[english]{babel}
\usepackage{tikz}
\usepackage{fancyhdr}

\theoremstyle{definition}
\newtheorem{definition}{Definition}[section]

\theoremstyle{remark}
\newtheorem{oss}[definition]{Remark}

\theoremstyle{plain}
\newtheorem{theorem}[definition]{Theorem}

\newtheorem{prop}[definition]{Proposition}
\newtheorem{lemma}[definition]{Lemma}
\newtheorem{cor}[definition]{Corollary}

\DeclareMathOperator{\Sym}{Sym}

\DeclareMathOperator{\Lin}{Lin}
\DeclareMathOperator{\Pol}{Pol}
\DeclareMathOperator{\Ker}{Ker}
\DeclareMathOperator{\Imag}{Im}

\title[The BRST cohomology complex of a matrix model]{The BRST cohomology and a generalized Lie algebra cohomology: analysis of a matrix model}

\author{Roberta A. Iseppi}
\address{Center for Quantum Geometry of Moduli Spaces, Aarhus University, Ny Munkegade 118, 8000 Aarhus, Denmark}
\email{roberta.iseppi@qgm.au.dk}

\date{\today}

\begin{document}

\begin{abstract}
This article is devoted to the analysis of the gauge-fixed BRST cohomology complex for a matrix model endowed with a {\small{$U(2)$}}-gauge symmetry. After a brief introduction on the BV construction and the gauge-fixing procedure in the setting of finite-dimensional gauge theories, we apply these constructions to  the model, with the purpose of explicitly determining and computing the corresponding gauge-fixed BRST cohomology groups. In the second part of this article,  we introduce a notion of {\em generalized Lie algebra cohomology}, which allows the gauge-fixed BRST cohomology complex for a new description, able to detect a {\em double complex structure}. 
\end{abstract}

\maketitle

\section{The BV construction and the gauge-fixing procedure: motivations}
\label{Section: intro}
The Batalin-Vilkovisky (BV) formalism (cf. \cite{BV1}, \cite{BV2}, \cite{BV3}) can be viewed as the end point of a long path, which had its motivation in the problem of quantizing gauge theories via the path integral approach (cf.\cite{Feynman}) and its origin in the introduction of the concept of {\em ghost field} by Faddeev and Popov in 1967 (cf. \cite{Faddeev-Popov}). The ghost fields are extra non-physical fields that are used to enlarge a gauge theory, with the aim of overcoming the issues created by the presence of local symmetries. Given a gauge theory {\small{$(X_{0}, S_{0})$}} for {\small{$X_{0}$}} an initial configuration space and {\small{$S_{0}: X_{0} \rightarrow \mathbb{R}$}} an action functional on {\small{$X_{0}$}} that is invariant under the action of a gauge group {\small{$\mathcal{G}$}}, the BV construction can be seen as a method to determine a new extended pair {\small{$(\widetilde{X}, \widetilde{S})$}}
$$\begin{array}{ccc}
(X_{0}, S_{0}) & -------\rightarrow  &(\widetilde{X}, \widetilde{S})\\
\mbox{\small{initial gauge theory}} & \mbox{\tiny{BV construction}} & \mbox{\small{extended theory}}
\end{array}
$$
where the {\em extended configuration space} {\small{$\widetilde{X}$}} is obtained as extension of the initial configuration space {\small{$X_{0}$}} with {\em ghost/anti-ghost fields}:
$$\widetilde{X} = X_{0} \cup \{ \mbox{ghost/anti-ghost fields} \},$$
and the {\em extended action} {\small{$\widetilde{S}$}} is defined by adding extra terms depending on the ghost/anti-ghost fields to the initial action {\small{$S_{0}$}}:
$$\widetilde{S} = S_{0} + \mbox{terms depending on ghost/anti-ghost fields}.$$
Further conditions have to be imposed on the pair {\small{$(\widetilde{X}, \widetilde{S})$}} to be a proper extension of the initial gauge theory {\small{$(X_{0}, S_{0})$}}. As a consequence, every properly-extended theory {\small{$(\widetilde{X}, \widetilde{S})$}} naturally induces a so-called BRST cohomology complex, first discovered by Becchi, Rouet, Stora \cite{BRS}, \cite{BRS3} and, independently, by Tyutin \cite{T}. However, after having performed this BV construction, the action {\small{$\widetilde{S}$}} still turns out to be written in a form that is not appropriate for an analysis of the theory through methods coming from perturbation theory. Indeed, the action {\small{$\widetilde{S}$}} contains anti-fields/anti-ghost fields, which need to be eliminated before computing amplitudes and {\small{$S$}}-matrix elements. Hence, a {\em gauge-fixing procedure} has to be performed to remove the anti-fields/anti-ghost fields both from {\small{$\widetilde{X}$}} and {\small{$\widetilde{S}$}}. As recalled in Section \ref{Section: The gauge fixing and the BRST cohomology complex}, despite of a gauge-fixing procedure being based on the choice of a suitable {\em gauge-fixing fermion} {\small{$\Psi$}}, the construction is implemented in such a way that all the physically-relevant quantities are independent of the choice of {\small{$\Psi$}}. However, after a gauge-fixing procedure is applied, a natural question arises: does the gauge-fixed theory {\small{$(\widetilde{X}, \widetilde{S})|_{\Psi}$}} still induce a cohomology complex? The response is that not only there exists a gauge-fixed version of the BRST cohomology complex (at least on-shell), but also its  cohomology groups allow to recover interesting information on the initial gauge theory {\small{$(X_{0}, S_{0})$}}. For example, the degree-{\small{$0$}} gauge-fixed BRST cohomology group {\small{$H^{0}(\widetilde{X}|_{\Psi}, d_{\widetilde{S}}|_{\Psi}) $}} describes the space of classical observables of {\small{$(X_{0}, S_{0})$}}, that is, the space of regular functions on {\small{$X_{0}$}} which are invariant under the action of the gauge group {\small{$\mathcal{G}$}}:
$$H^{0}(\widetilde{X}|_{\Psi}, d_{\widetilde{S}}|_{\Psi}) = \{ \mbox{ Classical observables of the initial gauge theory } (X_{0}, S_{0})\ \}.$$
In the light of the very interesting role played by the gauge-fixed BRST cohomology complex, we devote this article to a complete analysis of this cohomology complex for a finite-dimensional gauge theory {\small{$(X_{0}, S_{0})$}} endowed with a {\small{$U(2)$}}-gauge symmetry. Indeed, even though neither the BV construction nor the gauge-fixing procedure are strictly required in the finite-dimensional setting, it turns out to be a surprisingly rich context for rigorously analyzing both these constructions. \\ 
\\
More in detail, Section \ref{Section: The gauge fixing and the BRST cohomology complex} is devoted to a brief description of the gauge-fixing procedure in the finite-dimensional setting and to the statement of the notion of {\em gauge-fixed BRST cohomology complex}. In Section \ref{section: Application to a $U(2)$-matrix model} we apply the aforementioned construction to a matrix model {\small{$(X_{0}, S_{0})$}} endowed with a {\small{$U(2)$}}-gauge symmetry, arriving to determine and compute its gauge-fixed BRST cohomology groups. Finally, the notion of {\emph{generalized Lie algebra cohomology complex}} is introduced in Section \ref{Section: A generalized notion of Lie algebra cohomology}
while our main result is presented in Section \ref{Section: BRST and generalized Lie algebra cohomology}, where the gauge-fixed BRST cohomology complex of our model is reformulated in terms of this new cohomology complex. This approach to the analysis of the gauge-fixed BRST complex allows a clearer understanding of its intrinsic structure by detecting a {\em double complex structure}, not visible at the level of the BRST complex. This article ends with Section \ref{Sect: Conclusions}, where we insert this original construction in a larger perspective, providing outlooks for its application to other models, that is, for finite dimensional gauge theories with a larger gauge symmetry group.\\
\\
\noindent
{\em Acknowledgments:} the research presented in this article originates from some preliminary results obtained by the author when supported by Netherlands Organization for Scientific Research (NWO), through Vrije Competitie (pro\-ject number 613.000.910). The final stages of writing this article and the main results incorporated in it have received funding from the European Union Horizon 2020 Research and Innovation Program, under the Marie Sklodowska-Curie grant agreement No 753962. The author also would like to thank Walter D. van Suijlekom for interesting remarks and inspiring conversations. 

\section{The gauge fixing procedure and the BRST cohomology complex}
\label{Section: The gauge fixing and the BRST cohomology complex}
This section is devoted to a brief presentation of the gauge-fixing procedure in the context of finite-dimensional gauge theories (cf. \cite{Fior}, \cite{Schw}). For completeness, we mention that this procedure has been developed also in the context of infinite-dimensional gauge theories and that different methods and approaches have been formalized to perform it. In what follows,  we present the fields/anti-fields approach, referring to \cite{AKSZ}, \cite{GPS}, \cite{Henneaux2} for a more exhaustive presentation of the subject in the infinite-dimensional framework.\\
\\
We start by recalling the notion of (finite-dimensional) {\em gauge theory}, which will be used through out the whole article.

\begin{definition}
Let {\small{$X_{0}$}} be a vector space over {\small{$\mathbb{R}$}}, {\small{$S_{0}$}} be a functional on {\small{$X_{0}$}}, {\small{$S_{0}: X_{0} \rightarrow \mathbb{R}$}}, and {\small{$\mathcal{G}$}} be a group acting on {\small{$X_{0}$}} through an action {\small{$F:\mathcal{G} \times X_{0} \rightarrow X_{0}.$}} Then the pair {\small{$(X_{0}, S_{0})$}} is a {\em gauge theory with gauge group {\small{$\mathcal{G}$}}} if it holds that
 $$S_{0}(F(g, \varphi)) = S_{0}(\varphi), \quad \quad \forall \varphi \in X_{0}, \ \forall g \in \mathcal{G}.$$
\end{definition}

Given a gauge theory {\small{$(X_{0}, S_{0})$}}, we will refer to {\small{$X_{0}$}} as the {\em configuration space}, an element {\small{$\varphi$}} in {\small{$X_{0}$}} will be called a {\em gauge field}, the functional {\small{$S_{0}$}} is known as the {\em action functional} and for the group {\small{$\mathcal{G}$}} we will use the  terminology {\em gauge group}.\\
\\
As already mentioned in Section \ref{Section: intro}, given an initial gauge theory {\small{$(X_{0}, S_{0})$}}, the BV construction is used to determine an extended theory {\small{$(\widetilde{X}, \widetilde{S})$}} via the introduction of extra non-physical fields, called {\em ghost fields} (cf. \cite{Faddeev-Popov}).

\begin{definition}
A {\em field/ghost field} {\small{$\varphi$}} is a graded variable characterized by two integers:
 $$\deg(\varphi) \in \mathbb{Z} \quad \mbox{ and } \quad \epsilon(\varphi) \in \{ 0, 1 \},  \quad \mbox{ with } \quad \deg(\varphi) = \epsilon(\varphi) \quad (\mbox{mod} \ \mathbb{Z}/2).$$
{\small{$\deg(\varphi)$}} is the {\em ghost degree}, while {\small{$\epsilon(\varphi)$}} is the {\em parity}, which distinguishes between the bosonic case, where {\small{$\epsilon(\varphi)=0$}} and {\small{$\varphi$}} behaves as a real variable, and the fermionic case, where {\small{$\epsilon(\varphi)=1$}} and {\small{$\varphi$}} behaves as a Grassmannian variable: 
$$\varphi \psi = - \psi \varphi, \quad \quad \mbox{ and }\quad\quad \varphi^{2} =0, \quad \quad  \mbox { if } \quad \epsilon(\varphi) = \epsilon(\psi) = 1.$$
The {\em anti-field/anti-ghost field} {\small{$\varphi^{*}$}} corresponding to a field/ghost field {\small{$\varphi$}} satisfies
$$\deg(\varphi^{*}) = - \deg(\varphi) -1, \quad \quad \mbox{ and } \quad \quad \epsilon(\varphi^{*}) = \epsilon(\varphi) +1, \quad (\mbox{mod} \ \mathbb{Z}/2).$$
\end{definition}
In what follows, the term {\em fields} is reserved to the initial fields in {\small{$X_{0}$}} while {\em ghost fields} is used to identify the extra fields introduced by the BV construction. Analogously, {\em anti-fields} is specifically used for the anti-particles corresponding to the initial fields while the {\em anti-ghost fields} are the ones corresponding to the ghost fields.

\begin{definition}
\label{def: extended theory}
Let the pair {\small{$(X_{0}, S_{0})$}} be a gauge theory. An {\em extended theory} associated to {\small{$(X_{0}, S_{0})$}} is a pair {\small{$(\widetilde{X}, \widetilde{S})$}} where the extended configuration space {\small{$\widetilde{X} = \oplus_{i \in \mathbb{Z}} [\widetilde{X}]^{i}$}} is a {\small{$\mathbb{Z}$}}-graded super-vector space suitable to be decomposed as
\begin{equation}
\label{def: extended conf. sp.}
\widetilde{X} \cong \mathcal{F} \oplus \mathcal{F}^{*}[1], \quad \quad \mbox{ with } \quad [\widetilde{X}]^{0} = X_{0}
\end{equation}
for {\small{$\mathcal{F} = \oplus_{i \geqslant 0} \mathcal{F}^{i}$}} a graded locally free {\small{$\mathcal{O}_{X_{0}}$}}-module with homogeneous components of finite rank and {\small{$\mathcal{O}_{X_{0}}$}} the algebra of regular functions on {\small{$X_{0}$}}. Moreover, concerning the extended action {\small{$\widetilde{S} \in [\mathcal{O}_{\widetilde{X}}]^{0}$}}, it is a real-valued regular function on {\small{$\widetilde{X}$}}, with {\small{$\widetilde{S}|_{X_{0}}=S_{0}$}}, {\small{$\widetilde{S}\neq S_{0}$}} and such that it solves the {\em classical master equation}, i.e., 
$$\{\widetilde{S}, \widetilde{S}\}=0,$$
where {\small{$\{ -, -\}$}} denotes the graded Poisson structure on the algebra {\small{$\mathcal{O}_{\widetilde{X}}$}}.
 \end{definition}

\begin{oss}
The condition imposed in Definition \ref{def: extended theory} of {\small{$\widetilde{X}$}} being a {\small{$\mathbb{Z}$}}-{\emph{graded super-vector space}} simply encodes the fact that ghost/anti-ghost fields of even degree have to be treated as real variables while they are Grassmannian variables if they have odd ghost degree. Moreover, for what concerns the decomposition condition required in \eqref{def: extended conf. sp.}, it enforces the prescription of the BV formalism of introducing all anti-fields/anti-ghost fields corresponding to the fields/ghost fields in the extended configuration space. Indeed, while {\small{$\mathcal{F}$}} describes the fields/ghost-fields content of {\small{$\widetilde{X}$}}, {\small{$\mathcal{F}^{*}[1]$}} determines the anti-fields/anti-ghost fields part, with {\small{$\mathcal{F}^{*}[1]$}} that denotes the shifted dual module of {\small{$\mathcal{F}$:}}
 $$\mathcal{F}^{*}[1] = \oplus_{i \in \mathbb{Z}} \big[ \mathcal{F}^{*}[1] \big]^{i} \quad \quad \mbox{ with } \quad \big[ \mathcal{F}^{*}[1] \big]^{i} = \big[ \mathcal{F}^{*} \big]^{i+1}.$$
Finally, as a consequence of {\small{$\widetilde{X}$}} being a {\small{$\mathbb{Z}$}}-graded super-vector space, the algebra {\small{$\mathcal{O}_{\widetilde{X}}$}} of real-valued regular functions defined on {\small{$\widetilde{X}$}} naturally inherits a {\small{$\mathbb{Z}$}}-graded algebra structure. In addition, {\small{$\mathcal{O}_{\widetilde{X}}$}} can also be endowed with a graded Poisson structure induced by bracket 
$$\{ -, -\}: [\mathcal{O}_{\widetilde{X}}]^{\bullet} \rightarrow [\mathcal{O}_{\widetilde{X}}]^{\bullet +1}$$ 
of degree {\small{$1$}}. This structure is completely determined by requiring that, on the generators of {\small{$\widetilde{X}$}}, it satisfies the following conditions
$$\big\{ \beta_{i}, \beta_{j}\big\}= 0 , \quad \quad \quad \big\{ \beta^{*}_{i}, \beta_{j}\big\} = \delta_{ij} \quad  \quad \mbox{ and } \quad \quad \big\{ \beta^{*}_{i}, \beta^{*}_{j}\big\}=0 $$  
for {\small{$\beta_{i} \in \mathcal{F}^{p}$}} and {\small{$\beta^{*}_{i} \in \big[ \mathcal{F}^{*}[1] \big]^{-p-1}$}}, {\small{$p \in \mathbb{Z}_{\geqslant 0}$}}, while its value on any other possible combination of fields/ghost fields/anti-fields and anti-ghost fields is equal to zero. Then, the definition of the bracket structure on the whole algebra {\small{$\mathcal{O}_{\widetilde{X}}$}} is obtained by enforcing the bracket being linear and graded Poisson. 
\end{oss}

Even though it is not explicitly required in the notion, in what follows we will consider  extended theories {\small{$(\widetilde{X}, \widetilde{S})$}} with {\em finite level of reducibility}. More explicitly, we assume that only a finite number of homogeneous components {\small{$[\widetilde{X}]^{i}$}} in the extended configuration space {\small{$\widetilde{X}$}} are non-trivial. Then, given an extended configuration space {\small{$\widetilde{X} \cong \mathcal{F} \oplus \mathcal{F}^{*}[1]$}}, for {\small{$\mathcal{F} = \oplus_{i=0}^{m} [\mathcal{F}]^{i}$}}, we will refer to this theory as a theory with level of reducibility {\small{$L$}} for {\small{$L:= m-1$}},  {\small{$L\in \mathbb{Z}_{\geqslant 0}$}}. 

\subsection{The gauge-fixing procedure}
\label{Subsect: The gauge-fixing procedure}
The first step that has to be taken to perform a gauge fixing procedure is the choice of a {\em gauge-fixing fermion}, whose definition we recall for completeness.

\begin{definition}
\label{def gauge fixing fermion}
Given an extended configuration space {\small{$\widetilde{X}$}}, a {\em gauge-fixing fermion} {\small{$\Psi$}} is a regular function {\small{$ \Psi \in [\mathcal{O}_{\mathcal{F}}]^{-1}$}}, that is, a regular function in fields/ghost fields of degree {\small{$-1$}} and odd parity. 
\end{definition}

Given an extended theory {\small{$(\widetilde{X}, \widetilde{S})$}} and a gauge-fixing fermion {\small{$\Psi$}}, a gauge-fixing procedure is performed to obtain a new pair {\small{$(\widetilde{X}, \widetilde{S})|_{\Psi}$}} where neither the {\em gauge-fixed configuration space} {\small{$\widetilde{X}|_{\Psi}$}} nor the {\em gauge-fixed action} {\small{$\widetilde{S}|_{\Psi}$}} depend on anti-fields/anti-ghost fields. At the level of the configuration space, this goal is reached by defining the gauge-fixed configuration space {\small{$\widetilde{X}|_{\Psi}$}} to be
$$\widetilde{X}|_{\Psi}:= \big[\mathcal{F} \oplus\mathcal{F}^{*}[1]\big]|_{\varphi_{i}^{*} = \frac{\partial \Psi}{\partial \varphi_{i}}}, $$
namely, the Lagrangian submanifold determined by imposing the collection of so-called {\em gauge-fixing conditions} {\small{$\varphi_{i}^{*} =\partial \Psi/\partial \varphi_{i}$}}, where every anti-field/anti-ghost field {\small{$\varphi^{*}_{i}$}} in {\small{$\mathcal{F}^{*}[1] $}} is replaced with the partial derivative of {\small{$\Psi$}} with respect to the corresponding field/ghost field {\small{$\varphi$}}. Similarly, also the {\em gauge-fixed action} {\small{$\widetilde{S}|_{\Psi}$}} is obtained by imposing the gauge-fixing condition:
$$\widetilde{S}(\varphi_{i}, \varphi^{*}_{i})|_{\Psi} := \widetilde{S}\Big(\varphi_{i}, \varphi_{i}^{*} =\frac{\partial \Psi}{\partial \varphi_{i}}\Big).$$
We will refer to the pair {\small{$(\widetilde{X}|_{\Psi}, \widetilde{S}|_{\Psi})$}} as the {\em gauge-fixed theory}.\\
\\
In order to ensure that the gauge-fixing procedure is well defined from a physical point of view, we need to verify that the physically relevant quantities do not depend on the choice of gauge-fixing fermion {\small{$\Psi$}}. This condition can be restated by requiring that, given a regular function g on {\small{$\widetilde{X}|_{\Psi}$}}, the quantity computed by integrals of the following type
$$\int_{\widetilde{X}|_{\psi}} g \ d(Vol_{\widetilde{X}|_{\Psi}}),$$
for {\small{$Vol_{\widetilde{X}|_{\Psi}}$}} a volume form on {\small{$\widetilde{X}|_{\psi}$}}, is invariant with respect to the explicit form of {\small{$\Psi$}}. The fulfilling of this requirement is a consequence of Schwarz's Theorem (cf. \cite{Schw}, \cite{Fior}), where it is proved that this class of integrals only depend on the homotopy class of the gauge-fixing fermion and not on its explicit form. 

\subsection{The gauge-fixed BRST cohomology}
\label{The gauge-fixed BRST cohomology}
After having briefly reviewed the gauge-fixing procedure, a natural question arises concerning the effect of this construction at the level of the BRST cohomology. Indeed, as mentioned in the introduction, given an extended theory {\small{$(\widetilde{X}, \widetilde{S})$}}, it naturally induces a so-called {\em classical BRST cohomology complex}, where the coboundary condition for the coboundary operator {\small{$d_{\widetilde{S}}$}} holds due to the extended action {\small{$\widetilde{S}$}} solving the classical master equation.  However, what is the effect at the level of the corresponding BRST cohomology complex of having performed the gauge-fixing procedure on the pair {\small{$(\widetilde{X}, \widetilde{S})$}}? Is there a residual BRST-symmetry on {\small{$(\widetilde{X}, \widetilde{S})|_{\Psi}$}} that induces a gauge-fixed BRST cohomology complex? This section is devoted to recall the answer to this question (cf. \cite{AKSZ}) as well as the notions of these two cohomology complexes we are interested in: the {\emph{classical BRST complex}} and the \emph{gauge-fixed BRST complex}.  

\begin{definition}
\label{definition classical BRST complex}
Given an extended theory {\small{$(\widetilde{X}, \widetilde{S})$}}, the induced {\em classical BRST cohomology complex} is a cohomology complex whose cochain spaces {\small{$\mathcal{C}^{i}(\widetilde{X}, d_{\widetilde{S}})$}} and coboundary operator {\small{$d_{\widetilde{S}}$}} are defined as follows, respectively: 
$$\mathcal{C}^{i}(\widetilde{X}, d_{\widetilde{S}}) := [Sym_{\mathcal{O}_{X_{0}}}(\widetilde{X})]^{i},  $$
for {\small{$ i \in \mathbb{Z}$}}, with {\small{$Sym_{\mathcal{O}_{X_{0}}}(\widetilde{X})$}} the {\small{$\mathbb{Z}$}}-graded symmetric algebra generated by {\small{$\widetilde{X}$}} on the ring {\small{$\mathcal{O}_{X_{0}}$}}, and 
$$d_{\widetilde{S}}: \mathcal{C}^{\bullet}(\widetilde{X}, d_{\widetilde{S}}) \rightarrow \mathcal{C}^{\bullet +1}(\widetilde{X}, d_{\widetilde{S}}), \quad \quad \mbox{ with } \quad \quad d_{\widetilde{S}}:= \{ \widetilde{S}, - \}$$ 
for {\small{$\{ - , - \}$}} denoting the Poisson bracket structure on {\small{$\mathcal{O}_{\widetilde{X}}$}}.
\end{definition}
The gauge-fixing procedure effects the classical BRST cohomology complex by restricting both the cochain spaces {\small{$\mathcal{C}^{i}(\widetilde{X}, d_{\widetilde{S}})$}} and the coboundary operator {\small{$d_{\widetilde{S}}$}} to the Lagrangian submanifold {\small{$\widetilde{X}|_{\Psi} $}} determined by the gauge-fixing conditions, as precisely stated in the following definition.

\begin{definition}
\label{definition of gauge-fixed BRST cohomology}
 Given an extended theory {\small{$(\widetilde{X}, \widetilde{S})$}}, with {\small{$\widetilde{X} = \mathcal{F} \oplus \mathcal{F}^{*}[1]$}} and {\small{$\widetilde{S} \in [\mathcal{O}_{\widetilde{X}}]^{0}$}}, together with a gauge-fixing fermion {\small{$\Psi \in [\mathcal{O}_{\mathcal{F}}]^{-1}$}}, the induced {\em gauge-fixed BRST complex} {\small{$(\mathcal{C}^{j}(\widetilde{X}|_{\Psi}, d_{\widetilde{S}}|_{\Psi}), d_{\widetilde{S}}|_{\Psi})$}} is a cohomology complex with
$$\mathcal{C}^{j}(\widetilde{X}|_{\Psi}, d_{\widetilde{S}}|_{\Psi}) := \big[ \Sym_{\mathcal{O}_{X_{0}}}(\widetilde{X}|_{\Psi})\big]^{j}, \quad \quad \mbox{ and } \quad \quad d_{\widetilde{S}}|_{\Psi} := \left.\big\{ \widetilde{S}, - \big\}\right|_{\widetilde{X}|_{\Psi}} $$
where {\small{$j \in \mathbb{Z}$}} and {\small{$\widetilde{X}|_{\Psi} \subset \widetilde{X}$}} is the Lagrangian submanifold defined by the gauge-fixing conditions {\small{$\{ \varphi_{i}^{*} = \partial \Psi/\partial \varphi_{i}\}.$}}
\end{definition}

Although the definition of the gauge-fixed BRST complex is simply the restriction of the one of classical BRST complex to the gauge-fixed configuration space {\small{$\widetilde{X}|_{\Psi}$}} and hence to the corresponding algebra {\small{$\mathcal{O}_{\widetilde{X}|_{\Psi}}$}} of regular functions on {\small{$\widetilde{X}|_{\Psi}$}}, it is not straightforward that the residual BRST complex, after gauge fixing, still defines a cohomology complex. In particular, being restricted to {\small{$\mathcal{O}_{\widetilde{X}|_{\Psi}}$}}, the operator {\small{$d_{\widetilde{S}}|_{\Psi}$}} might fail to satisfy the coboundary condition {\small{$d^{2}_{\widetilde{S}}|_{\Psi}(\varphi) = 0$}}, for all {\small{$\varphi \in \widetilde{X}|_{\Psi}$}}. Nevertheless, a direct computation allows us to determine if {\small{$d_{\widetilde{S}}|_{\Psi}$}} defines a coboundary operator just knowing the extended action {\small{$\widetilde{S}$}}. Explicitly, given a generic field/ghost field {\small{$\varphi_{i} \in \widetilde{X}|_{\Psi}$}}, it holds the following equality:
$${d^{2}_{\widetilde{S}}|_{\Psi}}(\varphi_{i}) = \left.\left[ \sum_{k} (-1)^{\epsilon(\varphi_{i})(1 + \epsilon(\varphi_{k}))} \frac{\partial \widetilde{S}|_{\Psi}}{\partial \varphi_{k}} \frac{\partial^{2} \widetilde{S}}{\partial \varphi_{k}^{*} \partial \varphi_{i}^{*}} \right]\right|_{\widetilde{X}|_{\Psi}},$$
 where {\small{$\epsilon(\varphi_{i})$}} denotes the parity (real or Grassmannian) of {\small{$\varphi_{i}$}}. Hence, while the condition of the theory being {\em on-shell}, i.e., of the equations of motion {\small{$ \partial\widetilde{S}|_{\Psi}/\partial \varphi_{i} = 0$}} been satisfied for all {\small{$\varphi_{i} \in \mathcal{F}$}}, automatically ensures that  {\small{$d_{\widetilde{S}}|_{\Psi}$}} defines a coboundary operator, the coboundary condition {\small{$d^{2}_{\widetilde{S}}|_{\Psi} = 0$}} might be satisfied also {\em off-shell}, depending on the explicit form of the action {\small{$\widetilde{S}$}}.
 
\subsection{Gauge-fixing auxiliary fields}
\label{Sect: Gauge-fixing auxiliary fields}
Before going to the analysis of the gauge-fixed BRST cohomology for a {\small{$U(2)$}}-model, there is still a technical aspect that has to be discussed. Indeed, in order to perform a gauge-fixing procedure on a pair {\small{$(\widetilde{X}, \widetilde{S})$}}, it is necessary to fix a gauge-fixing fermion {\small{$\Psi \in [\mathcal{O}_{\mathcal{F}}]^{-1}$}} for {\small{$\widetilde{X} = \mathcal{F} \oplus \mathcal{F}^{*}[1]$}}. However, there are cases when it is impossible to define a suitable {\small{$\Psi$}} as, for example, when the graded {\small{$\mathcal{O}_{X_{0}}$}}-module {\small{$\mathcal{F}$}}, which describes the fields/ghost fields content of {\small{$\widetilde{X}$}}, is only {\small{$\mathbb{Z}_{\geqslant 0}$}}-graded. It was to overcome this issue that Batalin and Vilkovisky introduced the notion of {\em auxiliary fields} (cf. \cite{BV1}, \cite{BV2}).

\begin{definition}
\label{auxiliary pair}
An {\em auxiliary pair} is a pair of fields {\small{$(B, h)$}} such that their ghost degrees and parities satisfy the following relations:
$$
\begin{array}{lr}
deg(h)= deg(B) +1; \quad \quad & \quad \quad \epsilon(h)= \epsilon(B) + 1 \mbox{ (mod 2) }. 
\end{array}
$$
\end{definition}

The auxiliary fields are used to face the problem of not having negatively-graded ghost fields in {\small{$\widetilde{X}$}}. Indeed, given an extended theory {\small{$(\widetilde{X}, \widetilde{S})$}}, it can be further extended via the introduction of an auxiliary pair {\small{$(B, h)$}}, where the free parameter {\small{$\deg(B)$}} can be fixed to be in {\small{$\mathbb{Z}_{<0}$}}. 

\begin{definition}
Given an extended theory {\small{$(\widetilde{X}, \widetilde{S})$}} and an auxiliary pair {\small{$(B, h)$}}, the corresponding {\em total theory} {\small{$(X_{tot}, S_{tot})$}} has a {\em total configuration space} {\small{$X_{tot}$}} defined as the {\small{$\mathbb{Z}$}}-graded super-vector space generated by {\small{$\widetilde{X}$}}, {\small{$(B, h)$}} and their corresponding antifields {\small{$(B^{*}, h^{*})$}}
$$X_{tot}:= \langle \widetilde{X}, B, h, B^*, h^* \rangle, $$
and a {\em total action} {\small{$S_{tot}$}} given by the following sum:
$$S_{tot}:= \widetilde{S} + S_{aux},\quad \quad \mbox{ where } \quad \quad S_{aux}:= h B^{*}.$$
\end{definition}

\begin{oss}
\label{Remark Poisson structure on total conf sp}
We notice that, by construction, also {\small{$X_{tot}$}} presents the symmetry between fields/ghost fields and anti-fields/anti-ghost fields content, being suitable for the following decomposition:
$$X_{tot} = \mathcal{F}_{tot} \oplus \mathcal{F}_{tot}^{*}[1],$$
where {\small{$\mathcal{F}_{tot}$}} is a {\small{$\mathbb{Z}$}}-graded finitely-generated {\small{$\mathcal{O}_{X_{0}}$}}-module, which describes the fields/ ghost fields content of {\small{$X_{tot}$}} whilst its shifted dual {\small{$\mathcal{F}_{tot}^{*}[1]$}} describes the anti-fields/anti-ghost fields content. Moreover, {\small{$X_{tot}$}} can be equipped with a graded Poisson structure by extending the one already defined on {\small{$\widetilde{X}$}} imposing the following conditions: 
$$\{ B^{*}, B\} = \{ h^{*}, h\} =1 \quad \quad \mbox{ and } \quad \quad \{ \varphi, \xi\}= 0,$$
for {\small{$\varphi$}} any generator in {\small{$\widetilde{X}$}} and {\small{$\xi$}} any auxiliary fields among {\small{$B$}}, {\small{$h$}} or their corresponding anti-fields {\small{$B^{*}$}} and  {\small{$h^{*}$}}, while the value of the bracket on any other possible combination of auxiliary fields and corresponding auxiliary antifields is  declared to be zero.  
A straight consequence of this way of equipping {\small{$X_{tot}$}} with a Poisson structure is that the total action {\small{$S_{tot}$}} solves the classical master equation on {\small{$\mathcal{O}_{X_{tot}}$}}:
$$\{ S_{tot}, S_{tot}\} =0.$$
Therefore, the theories {\small{$(\widetilde{X}, \widetilde{S})$}} and {\small{$(X_{tot}, S_{tot})$}} satisfy similar properties, the only difference lying in the fact that {\small{$X_{tot}$}} may also contain negatively graded fields.
\end{oss}

Thus it appears that our goal has been achieved: by enlarging our extended theory with the introduction of an auxiliary pair {\small{$(B, h)$}} such that {\small{$\deg(B)=-1$}}, any {\small{$\Psi = f B$}}, with {\small{$f$}} in {\small{$\mathcal{O}_{X_{0}}$}} defines a gauge-fixing fermion, since {\small{$\Psi \in [\mathcal{O}_{\mathcal{F}_{tot}}]^{-1}$}}.  However, in addition to the conditions already imposed by its definition, to describe a well-defined physical theory a gauge-fixing fermion should determine a gauge-fixed action {\small{$S_{tot}|_{\Psi}$}} which is a {\em proper solution} of the classical master equation (cf. \cite{GPS}). This is the reason why, depending on the gauge theory considered, it might be necessary to introduce more than one auxiliary pair. More precisely, as inductively proved by Batalin and Vilkovisky (cf. Theorem \ref{teorema auxiliary fields caso generale}, \cite{BV1}, \cite{BV2}), the minimal number of auxiliary pairs that have to be added  depends on the {\em level of reducibility} of the extended theory {\small{$(\widetilde{X}, \widetilde{S})$}} considered. 

\begin{definition}
\label{def: level of reducibility}
An extended theory {\small{$(\widetilde{X}, \widetilde{S})$}}, with {\small{$\widetilde{X}= \mathcal{F} \oplus \mathcal{F}^{*}[1]$}} for {\small{$\mathcal{F}$}} a {\small{$\mathbb{Z}_{\geqslant 0}$}}-graded and finitely-generated {\small{$\mathcal{O}_{X_{0}}$}}-module, is {\em reducible with level of reducibility} {\small{$L:= m -1\geqslant 1$}} if {\small{$\mathcal{F}= \bigoplus_{i=0}^{m}\mathcal{F}_{i}$}}. Otherwise, if {\small{$L=0$}}, the theory is called {\em irreducible}.
\end{definition}

\begin{theorem}
\label{teorema auxiliary fields caso generale}
Given an extended theory {\small{$(\widetilde{X}, \widetilde{S})$}} with level of reducibility L, in order to determine a total theory {\small{$(X_{tot}, S_{tot})$}} whose gauge-fixed action {\small{$S_{tot}|_{\Psi}$}} is a proper solution of the classical master equation, a collection of auxiliary pairs {\small{$\{ (B_{i}^{j}, h_{i}^{j}) \}$}}, with {\small{$i= 0, \dots, L$}}, {\small{$j=1,\dots, i+1$}}, has to be introduced, which is completely determined by imposing that 
$$deg(B_{i}^{j}) = j -i -2  \quad \mbox{ if j is odd}, \quad \quad deg(B_{i}^{j}) = i -j +1 \quad \mbox{ if j is even}.$$
\end{theorem}

To conclude, we remark that the reason for introducing the auxiliary fields in pairs satisfying the properties listed in Definition \ref{auxiliary pair} is to not modify the classical BRST cohomology complex defined by the extended theory {\small{$(\widetilde{X}, \widetilde{S})$}}. Indeed, the auxiliary pairs are irrelevant from a cohomological point of view since they determine {\em contractible pairs}. 

\begin{definition}
\label{def: contractible pair}
Let {\small{$V$}} be a {\small{$\mathbb{Z}$}}-graded space of generators for a cohomology complex {\small{$(\mathcal{C}^{\bullet}(V, d), d)$}} with 
$$\mathcal{C}^{\bullet}(V, d) = \Sym^{\bullet}_{R}(V) \quad \quad  \mbox{ and } \quad \quad d:\mathcal{C}^{\bullet}(V, d) \rightarrow \mathcal{C}^{\bullet +1}(V, d), $$
where {\small{$R$}} is a ring and {\small{$d$}} is a {\small{$1$}}-degree {\small{$R$}}-linear coboundary operator. If there exist a pair of generators {\small{$B$}}, {\small{$h \in V$}} such that 
$$d(B) = h, \quad \quad \quad d(h) = 0 \quad \quad \quad \mbox{and} \quad \quad \quad d(x) = d|_{V\setminus \{ B, h \}}(x) \,$$
for any generator {\small{$x\in V\setminus \{ B, h \}$}}, then {\small{$(B, h)$}} defines a {\em contractible pair} for the co\-ho\-mo\-lo\-gy complex {\small{$(\mathcal{C}^{\bullet}(V, d), d)$}}.
\end{definition}

\begin{oss}
The fact that auxiliary pairs determine contractible pairs for the classical BRST cohomology complex is an immediate consequence of how the coboundary operator {\small{$d_{S_{tot}}$}} is defined for the BRST cohomology complex induced by the total theory {\small{$(X_{tot}, S_{tot})$}}. Indeed, the operator {\small{$d_{S_{tot}}:= \{ S_{tot}, \ \}$}} depends on the generators in {\small{$\widetilde{X}$}} only through the extended action {\small{$\widetilde{S}$}} so {\small{$d_{S_{tot}} = d_{\widetilde{S}}$}} on {\small{$\mathcal{O}_{\widetilde{X}}$}} while the action of {\small{$d_{S_{tot}}$}} on the auxiliary fields is 
$$d_{S_{tot}}(B) = \{ S_{aux}, B \} = h, \quad \quad \mbox{ and } \quad \quad d_{S_{tot}}(h) = 0, $$
as required in the notion of contractible pair. 
 \end{oss}
 
 To conclude, each auxiliary pair {\small{$(B^{i}_{j}, h^{i}_{j})$}} introduced to properly implement the gauge-fixing procedure determines a contractible pair for the classical BRST cohomology complex. The result on the cohomological triviality of contractible pairs is recalled in the following theorem. For ideas on the proof a possible reference is \cite{BBH}.
 
\begin{theorem}
\label{theorem trivial pairs}
Given {\small{$V$}} a {\small{$\mathbb{Z}$}}-graded space of generators for a cohomology complex {\small{$(\mathcal{C}^{\bullet}(V, d), d)$}} and {\small{$(B, h)$}} a contractible pair for this complex, the cohomology complexes {\small{$\mathcal{C}^{\bullet}(V, d)$}} and {\small{$\mathcal{C}^{\bullet}(V\setminus \{ B, h \}, d|_{V\setminus \{ B, h \}})$}} are quasi-isomorphic, that is, the following isomorphism holds 
 $$\mathcal{H}^{k}(V, d) \simeq \mathcal{H}^{k}(V\setminus \{ B, h \}, d|_{V\setminus \{ B, h \}}),$$
 for {\small{$k \in \mathbb{Z}$.}} 
\end{theorem}

\section{Application to a $U(2)$-matrix model}
\label{section: Application to a $U(2)$-matrix model}
The gauge-fixing procedure previously described is now applied on an example, in order to explicitly analyzed the induced gauge-fixed BRST cohomology complex. The gauge theory {\small{$(X_{0}, S_{0})$}} we focus on is a matrix model endowed with a {\small{$U(2)$}}-gauge symmetry, where the initial configuration space {\small{$X_{0}$}} is defined to be
$$X_{0} = \{ M \in M_{2}(\mathbb{C}): M^{*} = M \}$$
and the initial action functional {\small{$S_{0}:X_{0} \rightarrow \mathbb{R}$}} is supposed to be  invariant under the adjoint action of the unitary group {\small{$U(2)$}}, i.e., it satisfies the following equality
$$S_{0}[M] = S_{0}[UMU^*]$$
for all {\small{$M \in M_{2}(\mathbb{C})$}}, {\small{$U \in U(2)$}}. A more explicit representation of this model can be provided by fixing as basis of {\small{$X_{0}$}} the one given by the Pauli's matrices {\small{$\sigma_{1}, \sigma_{2}, \sigma_{3}$}}, together with the identity matrix {\small{$\sigma_{4}= Id$}}:
\begin{equation}
\label{basis}
\sigma_{1}= \begin{pmatrix}
		    0 & 1\\
		    1 & 0
	\end{pmatrix}
, \quad 
\sigma_{2}= \begin{pmatrix}
		    0 & -i\\
		    i & 0
		\end{pmatrix}
, \quad 
\sigma_{3}= \begin{pmatrix}
		    1 & 0\\
		    0 & -1
		\end{pmatrix}	 
, \quad
\sigma_{4}= \begin{pmatrix}
		    1 & 0\\
		    0 & 1
		\end{pmatrix}.
\end{equation}
Hence, our model of interest can be described as follows:
\begin{equation}
\label{eq S0}
X_{0} \cong \langle M_{1}, \dots, M_{4} \rangle, \quad \quad \mbox{ and } \quad \quad S_{0}= \sum_{k=0}^{r} \mbox{ }(M_{1}^2 + M_{2}^2 + M_{3}^2)^k \mbox{ } g_{k}(M_4),
\end{equation}
where {\small{$M_{a}$, $a=1, \dots, 4$}}, are independent real fields which generate {\small{$X_{0}$}} as vector space and {\small{$g_{k}(M_{4})$}} are polynomials in {\small{$\Pol_{\mathbb{R}}(M_4)$}}. Because we have already investigated how to construct an extended theory {\small{$(\widetilde{X}, \widetilde{S})$}} for this gauge theory {\small{$(X_{0}, S_{0})$}} (cf. Theorems {\small{$4.2$, $4.3$}} in \cite{primo_articolo}), after briefly recalling the main result in Theorem \ref{Theorem: extended theory U(2)}, we conclude the construction by first further enlarging the extended configuration space {\small{$\widetilde{X}$}} via the introduction of the necessary auxiliary pairs, then performing the gauge-fixing procedure and, finally, exhaustively analyzing the induced gauge-fixed BRST cohomology complex.

\begin{theorem}
\label{Theorem: extended theory U(2)}
Given a gauge theory {\small{$(X_{0}, S_{0})$}}, with {\small{$X_{0}\cong \mathbb{A}_{\mathbb{R}}^{4}$}} and {\small{$S_{0} \in \mathcal{O}_{X_{0}}$}} of the form \eqref{eq S0}, if {\small{$GCD(\partial_{a}S_{0}) = 1$}}, for {\small{$a=1, \dots, 4$}}, then the minimally-extended configuration space {\small{$\widetilde{X}$}} is the following {\small{$\mathbb{Z}$}}-graded super-vector space
\begin{equation}
\label{eq: extended configuration space}
\widetilde{X} = \langle E^{*} \rangle_{-3} \oplus \langle C^{*}_{1}, \cdots, C^{*}_{3}\rangle_{-2} \oplus \langle M^*_{1}, \dots, M^*_{4}\rangle_{-1} \oplus X_{0} \oplus \langle C_{1}, \cdots, C_{3} \rangle_{1} \oplus \langle E\rangle_{2}.
\end{equation} 
Moreover, the most general solution of the classical master equation on {\small{$\widetilde{X}$}} that is linear in the anti-fields, of at most degree {\small{$2$}} in the ghost fields and with coefficients in {\small{$\mathcal{O}_{X_{0}}$}} is the following one:
{\small{\begin{multline}
\label{eq: generic extended action}
\widetilde{S}= S_{0} + \sum_{i, j, k} \epsilon_{ijk}\alpha_{k}M_{i}^{*}M_{j}C_{k} + \sum_{i, j,k} C_{i}^{*}\big[\tfrac{\alpha_{j}\alpha_{k}}{2\alpha_{i}} (\beta\alpha_{i}M_{i}E + \epsilon_{ijk}C_{j}C_{k}) \\
+ M_{i}T \big( \sum_{a, b,c} \epsilon_{abc}\tfrac{\alpha_{b}\alpha_{c}}{2\alpha_{i}}M_{a}C_{b}C_{c}\big)\big]
\end{multline}
}}
where {\small{$\alpha_{i}, \beta\in {\small{\mathbb{R}\backslash \left\lbrace 0 \right\rbrace}}$,}} {\small{$T \in \mathcal{O}_{X_{0}}$}}, and {\small{$\epsilon_{ijk}$ ($\epsilon_{abc}$)}} is the totally anti-symmetric tensor in three indices {\small{$i, j, k \in \{ 1, 2, 3 \}$ ($a, b, c \in \{ 1, 2,3\}$)}} with {\small{$\epsilon_{123}=1$}}. 
\end{theorem}
To simplify the upcoming computation, we consider as extended action {\small{$\widetilde{S}$}} the one obtained by choosing {\small{$\alpha_{i}=\beta =1$}} and {\small{$T=0$}} in \eqref{eq: generic extended action}. Hence, we are going to apply the gauge-fixing procedure to an extended theory {\small{$(\widetilde{X}, \widetilde{S})$}} where {\small{$\widetilde{X} = \mathcal{F} \oplus \mathcal{F}^{*}[1]$}} is of the form describe in \eqref{eq: extended configuration space} and {\small{$\widetilde{S}$}} takes the following explicit form:
\begin{equation}
\label{eq: extended action conti}
\widetilde{S}= S_{0} + \sum_{i, j, k} \epsilon_{ijk}M_{i}^{*}M_{j}C_{k} + \sum_{i, j,k} C_{i}^{*}(M_{i}E + \epsilon_{ijk}C_{j}C_{k}). 
\end{equation}
Once the extended theory {\small{$(\widetilde{X}, \widetilde{S})$}} has been determined, the first step that has to the taken in order to implement a gauge-fixing procedure is the introduction of the auxiliary pairs. Indeed, as it can be immediately deduced by the description of {\small{$\widetilde{X}$}} in \eqref{eq: extended configuration space}, at this point it is impossible to define a suitable gauge-fixing fermion {\small{$\Psi \in [\mathcal{O}_{\mathcal{F}}]^{-1}$}}due to the absence of negatively-graded fields/ghost fields.\\
\\ To determine type and number of the required auxiliary pairs we first of all notice that, accordingly to Definition \ref{def: level of reducibility}, the extended theory {\small{$(\widetilde{X}, \widetilde{S})$}} for our model has level of reducibility {\small{$L=1$}}. Therefore, by applying Theorem \ref{teorema auxiliary fields caso generale}, we conclude that the extended configuration space {\small{$\widetilde{X}$}} has to be further enlarged by the introduction of {\em five auxiliary pairs}, three of which correspond to the ghost fields {\small{$C_{i}$}} of degree {\small{$1$}} and hence are defined as follows:
$$(B_{i}, h_{i})_{i=1, 2, 3} \quad \quad \mbox{ with } \quad  \quad \deg(B_{i}) =1$$ 
while the remaining two auxiliary pairs correspond to the ghost field {\small{$E$}} of degree {\small{$2$}} and hence are defined to be
$$(A_{m}, k_{m})_{m=1, 2} \quad \quad \mbox{ with } \quad \deg(A_{1}) =-2 \quad \mbox{ and } \quad \deg(A_{2}) =0.$$
Thus the total theory {\small{$(X_{tot}, S_{tot})$}} obtained by enlarging the extended theory {\small{$(\widetilde{X}, \widetilde{S})$}} via the introduction of the aforementioned auxiliary pairs has a  total configuration space suitable to be decomposed as the following direct sum {\small{$X_{tot} = \mathcal{F}_{tot} \oplus \mathcal{F}^{*}_{tot}[1]$}} with 
\begin{equation}
\label{eq: F tot}
\mathcal{F}_{tot} = \langle A_{1} \rangle_{-2} \oplus \langle B_{i}, k_{1} \rangle_{-1} \oplus \langle M_{a}, h_{i}, A_{2} \rangle_{0} \oplus \langle C_{i},   k_{2} \rangle_{1}\oplus \langle E\rangle_{2}, 
\end{equation}
for {\small{$a=1, \dots, 4$}}, {\small{$i=1, 2, 3$}}. Concerning the total action {\small{$S_{tot}$}}, it is obtained by adding to the extended action {\small{$\widetilde{S}$}} the following auxiliary summand:
\begin{equation}
\label{eq: S aux}
S_{aux} =  \sum_{i=1}^{3} B^{*}_{i} h_{i} +\sum_{m=1, 2} A^{*}_{m}k_{m}.
\end{equation}
Once the total theory {\small{$(X_{tot}, S_{tot})$}} has been constructed, a gauge-fixing fermion {\small{$\Psi \in [\mathcal{O}_{\mathcal{F}_{tot}}]^{-1}$}} has to be chosen. However, as explained in more details in the upcoming section, due to the properties of our model of interest, the explicit form of the gauge-fixing fermion will not play a role in the definition of the gauge-fixed BRST complex. Hence, we do not go into details in analyzing the best possible choice for the gauge-fixing fermion {\small{$\Psi$}} but we directly focus on the description of the gauge-fixed BRST complex, to which the next section is devoted. 

\subsection[The gauge-fixed BRST complex]{The gauge-fixed BRST cohomology of a $U(2)$-model}
Let {\small{$(X_{tot}, S_{tot})$}} be the total theory described in the previous section and corresponding to our {\small{$U(2)$}}-model of interest, with {\small{$X_{tot}= \mathcal{F}_{tot} \oplus \mathcal{F}_{tot}^{*}[1]$}}. Moreover, let us fix a gauge-fixing fermion {\small{$\Psi \in [\mathcal{O}_{\mathcal{F}_{tot}}]^{-1}$}} homotopically equivalent to the zero function. While the explicit form of {\small{$\Psi$}} in terms of the fields/ghost fields in {\small{$\mathcal{F}_{tot}$}} will not play any role, the condition of being  homotopically equivalent to the zero function is sufficient to ensures that the gauge-fixed configuration space {\small{$X_{tot}|_{\Psi}$}} coincides with {\small{$\mathcal{F}_{tot}$}}, namely the subspace of {\small{$X_{tot}$}} generated by fields/ghost fields. The consequence of this choice is a simplification of the structure of the gauge-fixed BRST complex, which will result quasi-isomorphic to a one-sided cohomology complex.\\
\\
Given the total theory {\small{$(X_{tot}, S_{tot})$}}, whose total configuration space {\small{$X_{tot}$}} has been described in \eqref{eq: F tot} and whose total action {\small{$S_{tot} = \widetilde{S} + S_{aux}$}} was explicitly written in \eqref{eq: extended action conti} and \eqref{eq: S aux}, let {\small{$\Psi \in [\mathcal{O}_{\mathcal{F}_{tot}}]^{-1}$}} be a gauge-fixing fermion satisfying the condition that {\small{$\Psi \equiv 0$}}. Then, according to Definition \ref{definition of gauge-fixed BRST cohomology}, the induced \emph{gauge-fixed BRST cohomology complex} {\small{$(\mathcal{C}^{\bullet}(X_{tot}|_{\Psi}, d_{S_{tot}}|_{\Psi}), d_{S_{tot}}|_{\Psi}))$}}, has
$$\mathcal{C}^{i}(X_{tot}|_{\Psi}, d_{S_{tot}}|_{\Psi})= [\Sym_{\mathcal{O}_{X_{0}}}(\mathcal{F}_{tot})]^{i}$$
with {\small{$i \in \mathbb{Z}$}},  as space of cochains and, as coboundary operator
$$d_{S_{tot}}|_{\Psi}: \mathcal{C}^{i}(X_{tot}|_{\Psi}, d_{S_{tot}}|_{\Psi}) \rightarrow \mathcal{C}^{i+1}(X_{tot}|_{\Psi}, d_{S_{tot}}|_{\Psi})$$
the operator uniquely determined by imposing it describes a {\small{$1$}}-degree linear and graded derivation that acts as follows on the generators of the cohomology complex:
\begin{equation}
\label{tot coboundary}
\begin{array}{ll}
\left\lbrace
\begin{array}{l}
d_{S_{tot}}|_{\Psi}(M_{l})= - \sum_{j,k} \epsilon_{ljk} M_{j}C_{k}\\
[2ex]
d_{S_{tot}}|_{\Psi}(C_{l})= \sum_{j,k} (M_{l}E + \epsilon_{ljk}C_{j}C_{k})\\
[2ex]
d_{S_{tot}}|_{\Psi}(E)=0.
\end{array}
\right.
&
\left\lbrace
\begin{array}{l}
d_{S_{tot}}|_{\Psi}( B_{i})= h_{i}\\
[.7ex]
d_{S_{tot}}|_{\Psi}( h_{i})= 0\\
[.7ex]
d_{S_{tot}}|_{\Psi}( A_{m}) = k_{m} \\
[.7ex]
d_{S_{tot}}|_{\Psi}( k_{m}) = 0.\\
\end{array}
\right.
\end{array}
 \end{equation} 

\begin{oss}
As already observed, assuming that the gauge-fixing fermion satisfies the condition of being homotopically equivalent to the zero function, i.e. {\small{$\Psi\equiv 0$}}, implies that
$$X_{tot}|_{\Psi}:= X_{tot}|_{\varphi^{*}_{i} = \frac{\partial \Psi}{\partial \varphi_{i}}} = \mathcal{F}_{tot},$$
for {\small{$\mathcal{F}_{tot}$}} the space of fields/ghost fields in {\small{$X_{tot}$}}. Moreover, because the total action {\small{$S_{tot}$}} is linear in the anti-fields/anti-ghost fields, already before the gauge-fixing procedure has been implemented, anti-fields/anti-ghost fields do not appear in the description of the action of the BRST coboundary operator {\small{$d_{S_{tot}}$}} on the fields/ghost fields. Hence, the imposition of the gauge-fixing condition does not have any effect at the level of the action of {\small{$d_{S_{tot}}$}} on all cochains defined on {\small{$X_{tot}|_{\Psi}$}}. Therefore, the particular gauge-fixing fermion {\small{$\Psi$}} chosen does not play an active role for this model. \end{oss}

Finally, from \eqref{tot coboundary}, it straightforwardly follows that the auxiliary pairs {\small{$(B_{i}, h_{i})$}}, for {\small{$i=1, 2, 3$}}, and {\small{$(A_{m}, k_{m})$}}, for {\small{$m=1, 2$}}, are all contractible pairs for the cohomology complex {\small{$(\mathcal{C}^{\bullet}(X_{tot}|_{\Psi}, d_{S_{tot}}|_{\Psi}), d_{S_{tot}}|_{\Psi})$}}. Hence, as immediate consequence of Theorem \ref{theorem trivial pairs} we can infer the following statement.

\begin{prop}
Let {\small{$(\widetilde{X}, \widetilde{S})$}} and {\small{$(X_{tot}, S_{tot})$}} be respectively the extended and the total theory associated to a {\small{$U(2)$}}-matrix model and explicitly described in \eqref{eq: extended configuration space}, \eqref{eq: F tot}, \eqref{eq: extended action conti} and \eqref{eq: S aux}, with {\small{$X_{tot} = \mathcal{F}_{tot} \oplus \mathcal{F}^{*}_{tot}[1]$}} and {\small{$S_{tot} = \widetilde{S} + S_{aux}$}}. Then, given a gauge-fixing fermion {\small{$\Psi \in [\mathcal{O}_{\mathcal{F}_{tot}}^{-1}]$}} such that {\small{$\Psi \equiv 0$}}, the induced gauge-fixed BRST complex {\small{$(\mathcal{C}^{\bullet}(X_{tot}|_{\Psi}, d_{S_{tot}}|_{\Psi}), d_{S_{tot}}|_{\Psi}))$}} is quasi-isomorphic to the complex {\small{$(\mathcal{C}^{\bullet}(\widetilde{X}|_{\Psi}, d_{\widetilde{S}}|_{\Psi}), d_{\widetilde{S}}|_{\Psi})$}}, with cochain spaces
$$\mathcal{C}^{i}(\widetilde{X}|_{\Psi}, d_{\widetilde{S}}|_{\Psi}): = [\Sym_{\mathcal{O}_{X_{0}}}(\mathcal{F})]^{i}, $$
for {\small{$i \in \mathbb{Z}_{\geqslant 0}$}} and {\small{$\mathcal{F}$}} the {\small{$\mathbb{Z}$}}-graded {\small{$\mathcal{O}_{X_{0}}$}}-module such that {\small{$\widetilde{X} = \mathcal{F} \oplus \mathcal{F}^{*}[1]$}},  and  coboundary operator {\small{$d_{\widetilde{S}}|_{\Psi}$}} the 1-degree operator defined to be the following restriction
$$d_{\widetilde{S}}|_{\Psi}: \mathcal{C}^{i}(\widetilde{X}|_{\Psi}, d_{\widetilde{S}}) \rightarrow \mathcal{C}^{i+1}(\widetilde{X}|_{\Psi}, d_{\widetilde{S}}), \quad \quad \mbox{ with } \quad \quad d_{\widetilde{S}}:= (d_{S_{tot}}|_{\Psi})|_{\mathcal{O}_{\mathcal{F}}}.$$
In other words, the following isomorphism holds for every {\small{$i \in \mathbb{Z}_{\geqslant 0}$}}:
$$H^{i}(X_{tot}|_{\Psi}, d_{S_{tot}}|_{\Psi}) \simeq H^{i}(\widetilde{X}|_{\Psi}, d_{\widetilde{S}}|_{\Psi}).$$
\end{prop}

As final remark we notice that the above proposition and the fact that the  {\small{$\mathcal{O}_{X_{0}}$}}-module {\small{$\mathcal{F}$}} is {\small{$\mathbb{Z}_{\geqslant 0}$}}-graded implies that the gauge-fixed BRST complex is quasi-isomorphic to a one-sided complex: this fact will simplify the explicit computation of those cohomology groups, to which the following section is devoted. 

\subsection[The BRST cohomology groups]{The gauge-fixed BRST groups: an explicit computation}
\label{Computation of the BRST cohomology groups for the model}
Before facing the problem of explicitly computing the cohomology groups determined by the cohomology complex {\small{$(\mathcal{C}^{\bullet}(\widetilde{X}|_{\Psi}, d_{\widetilde{S}}|_{\Psi}), d_{\widetilde{S}}|_{\Psi})$}} constructed in the above section, we recall that {\small{$\widetilde{X}|_{\Psi} = \mathcal{F}$}} is defined to be the following {\small{$\mathcal{O}_{X_{0}}$}}-module  
$$\mathcal{F} = \langle M_{1}, \dots, M_{4}\rangle_{0} \oplus \langle C_{1}, C_{2}, C_{3} \rangle_{1} \oplus \langle E \rangle_{2},$$
where the generators {\small{$C_{i}$}} have to be treated as Grassmannian variables while {\small{$E$}} as well as the generators {\small{$M_{a}$}} are real variables. As immediate consequence of the parity assigned to the generators of {\small{$\mathcal{F}$}}, we deduce the following proposition.

\begin{prop}
\label{isom x gradi maggiori di 4}
The cohomology groups determined by the cohomology complex \\{\small{$(\mathcal{C}^{\bullet}(\widetilde{X}|_{\Psi}, d_{\widetilde{S}}|_{\Psi}), d_{\widetilde{S}}|_{\Psi})$}} are 2-periodic, that is, the following isomorphisms are satisfied 
$$H^{2i+1}(\widetilde{X}|_{\Psi}, d_{\widetilde{S}}|_{\Psi}) \simeq H^{3}(\widetilde{X}|_{\Psi}, d_{\widetilde{S}}|_{\Psi})\quad  \mbox{ and }  \quad H^{2i+2}(\widetilde{X}|_{\Psi}, d_{\widetilde{S}}|_{\Psi}) \simeq H^{4}(\widetilde{X}|_{\Psi}, d_{\widetilde{S}}|_{\Psi}),$$
for all {\small{$i \geqslant 1$.}}
\end{prop}

\begin{proof}
The statement follows straightforwardly by noticing that already at the level of cocycles it holds that
$$\Ker(d^{2i+2}_{\widetilde{S}}|_{\Psi}) \cong \Ker(d^{2}_{\widetilde{S}}|_{\Psi}) \cdot E^{i} \quad \quad \mbox{ and } \quad \quad \Ker(d^{2i+3}_{\widetilde{S}}|_{\Psi}) \cong \Ker(d^{3}_{\widetilde{S}}|_{\Psi}) \cdot E^{i}$$
and similarly for the coboundary elements:
$$\Imag(d^{2i+4}_{\widetilde{S}}|_{\Psi}) \cong \Imag(d^{4}_{\widetilde{S}}|_{\Psi}) \cdot E^{i} \quad \quad \mbox{ and } \quad \quad \Imag(d^{2i+3}_{\widetilde{S}}|_{\Psi}) \cong \Imag(d^{3}_{\widetilde{S}}|_{\Psi}) \cdot E^{i},$$
for {\small{$i\geqslant 0$}}. 
\end{proof}

Hence, in order to completely determine the cohomology groups defined by the above complex it is enough to only consider and compute them up to degree {\small{$4$}}. Our goal is accomplished in the following theorem, the only exception being the degree {\small{$1$}} cohomology group. The complexity of its explicit computation will be overcome in Section 
\ref{Relation between the cohomology groups}, where the reformulation of the gauge-fixed BRST complex in terms of a generalized Lie algebra complex will allow to complete the description of these groups.

\begin{theorem}
\label{teorema con calcolo gruppi BRST per U(2)}
The gauge-fixed BRST cohomology complex {\small{$(\mathcal{C}^{\bullet}(\widetilde{X}|_{\Psi}, d_{\widetilde{S}}|_{\Psi}), d_{\widetilde{S}}|_{\Psi})$}} associated to a {\small{$U(2)$}}-matrix model determines the following cohomology groups:
\begin{enumerate}[$\blacktriangleright$]
 \item {\small{$H^{0}(\widetilde{X}|_{\Psi}, d_{\widetilde{S}}|_{\Psi}) = \{ \sum_{k=0}^{r} g_{k}(M_{4})(M_{1}^{2}+ M_{2}^{2} + M_{3}^{2})^{k}: g_{k} \in \Pol_{\mathbb{R}}(M_{4}), r \in 
 \mathbb{N} \}; \vspace{1mm}$}}
 \item {\small{$H^{2}(\widetilde{X}|_{\Psi}, d_{\widetilde{S}}|_{\Psi}) = \{ Q\sum_{i; j<k} \epsilon_{ijk} M_{i}C_{j}C_{k}: Q \in \mathcal{O}_{X_{0}}\}  \oplus \Pol_{\mathbb{R}}(M_{4})E; \vspace{1mm}$}}  
 \item in any odd degree {\small{$2q+1$}}, with {\small{$q\geqslant 1$}}, {\small{$H^{2q+1}(\widetilde{X}|_{\Psi}, d_{\widetilde{S}}|_{\Psi})= 0 ;\vspace{1mm}$}}
 \item in any even degree {\small{$2q$}}, with {\small{$q\geqslant 2$}}, 
$$H^{2q}(\widetilde{X}|_{\Psi}, d_{\widetilde{S}}|_{\Psi}) =  \Pol_{\mathbb{R}}(M_{4})E^{q}.$$ 
\end{enumerate}
\end{theorem}

\begin{proof}
The claim for case of {\em degree {\small{$0$}}} follows from the gauge-fixed BRST complex defining a one-sided cohomology, the independence of the variables {\small{$C_{i}$}} and the fact that, to be a {\small{$0$}}-degree cocycle, a polynomial {\small{$f \in \mathcal{O}_{X_{0}}$}} has to satisfy the conditions:
$$\partial_{i} f = M_{i}g,$$
for {\small{$i= 1, 2,3$}} and {\small{$g \in  \mathcal{O}_{X_{0}}$}}. 
\\
To prove the part of the statement concerning the cohomology group of {\em degree {\small{$2$}}}, we start by considering a generic cochain {\small{$\varphi$}} in {\small{$\mathcal{C}^{2}(\widetilde{X}|_{\Psi}, d_{\widetilde{S}}|_{\Psi})$}},  
$$\varphi= \sum_{j<k}g_{jk}C_{j}C_{k}+ hE,$$
with {\small{$g_{jk}$}}, {\small{$h \in \mathcal{O}_{X_{0}}$}} and {\small{$j, k =1,2 ,3$}}. By imposing the cocycle condition on {\small{$\varphi$}} we obtain the following result:
$$  \Ker(d^{2}_{\widetilde{S}}|_{\Psi}) \simeq \Big\{ 
 \varphi =   \sum_{j<k} \epsilon_{ijk} (M_{i}P-\partial_{i}h)C_{j}C_{k} + hE,   \mbox{ with }P, h \in \mathcal{O}_{X_{0}}
\Big\}.
 $$
However, since any element {\small{$\varphi \in \Ker(d^{2}_{\widetilde{S}}|_{\Psi})$}} can be uniquely written as
 $$\varphi = Q \sum_{j<k} \epsilon_{ijk} M_{i}C_{j}C_{k} + d_{\widetilde{S}}|_{\Psi}(\beta) + h_{0}E\vspace{-5mm}$$
 $$ \mbox{ for } \quad\quad h = \sum_{i=1}^{3} M_{i}A_{i} + h_{0}, \quad \quad  \beta= \sum_{i=1}^{3} A_{i}C_{i}, \quad \quad Q = - P + \sum_{i=1}^{3}\partial_{i}A_{i}.$$
with {\small{$A_{i}$}} elements in {\small{$\mathcal{O}_{X_{0}}$}} and {\small{$h_{0} \in \Pol_{\mathbb{R}}(M_{4})$}}, we deduce that the space of cocycles of degree {\small{$2$}} can be decomposed in the following direct sum:
\begin{equation}
 \label{Z^2(X, d_S)}
\Ker(d^{2}_{\widetilde{S}}|_{\Psi}) = \Big\{ Q\sum_{j<k} \epsilon_{ijk} M_{i}C_{j}C_{k}: Q \in \mathcal{O}_{X_{0}}\Big\} \oplus Im(d^{1}_{\widetilde{S}}|_{\Psi}) \oplus \Pol_{\mathbb{R}}(M_{4})E.
\end{equation} 
from which the claimed description of {\small{$H^{2}(\widetilde{X}|_{\Psi}, d_{\widetilde{S}}|_{\Psi})$}}  follows immediately.\\
Because of Proposition \ref{isom x gradi maggiori di 4}, to demonstrate the claim in any {\emph{odd degree {\small{$2q+1$}}}} it is enough to compute {\small{$H^{3}(\widetilde{X}|_{\Psi}, d_{\widetilde{S}}|_{\Psi})$}}. Given a generic cochain {\small{$\varphi$}} in {\small{$\mathcal{C}^{3}(\widetilde{X}|_{\Psi}, d_{\widetilde{S}}|_{\Psi})$}},
$$\varphi= f C_{1}C_{2}C_{3} + \sum_{i=1}^{3}g_{i}C_{i}E$$
for {\small{$f$, $g_{i} \in \mathcal{O}_{X_{0}}$}}, the enforcement of the cocycle condition implies that the polynomials {\small{$f$}} and {\small{$g_{i}$}} have to satisfy the following equalities:
$$ \sum_{i=1}^{3}M_{i}g_{i} =0 \quad \quad \mbox{ and } \quad \quad  f= \sum_{i=1}^{3}\partial_{i}g_{i}.
$$ 
As a consequence, {\small{$g_{i}= \sum_{j, k} \epsilon_{ijk}M_{j}P_{k},$}} for some {\small{$P_{k} \in \mathcal{O}_{X_{0}}$}}, from which we deduce that
$$\Ker(d^{3}_{\widetilde{S}}|_{\Psi}) = \bigg\{ \Big(\sum_{i,j,k} \epsilon_{ijk} M_{i}\partial_{j}P_{k}\Big)C_{1}C_{2}C_{3} + \Big(\sum_{i,j,k} \epsilon_{ijk}M_{i}P_{j}C_{k}\Big) E \bigg\}.$$
However, since any cocycle {\small{$ \varphi$}} in {\small{$\Ker(d^{3}_{\widetilde{S}}|_{\Psi})$}} can be viewed as a coboundary element {\small{$\varphi=d_{\widetilde{S}}|_{\Psi}(\beta)$}} for {\small{$\beta \in \mathcal{C}^{2}(\widetilde{X}|_{\Psi}, d_{\widetilde{S}}|_{\Psi})$}} defined to be 
$$\beta= - \sum_{j<k}\epsilon_{ijk} P_{i} C_{j}C_{k},$$
the claimed triviality of the cohomology group of degree {\small{$3$}} follows straightforwardly.\\
 To conclude the proof of the theorem, we still have to consider the cohomology groups of {\emph{even degree {\small{$2q$}}}}, for {\small{$q>1$}}. In view of Proposition \ref{isom x gradi maggiori di 4}, we deduce that the space of cocycles of degree {\small{$2q$}} satisfies the following isomorphism:
$$\Ker(d^{2q}_{\widetilde{S}}|_{\Psi})\simeq \bigg[\Big\{ Q\sum_{j<k} \epsilon_{ijk} M_{i}C_{j}C_{k}: Q \in \mathcal{O}_{X_{0}}\Big\} \oplus \Imag(d^{1}_{\widetilde{S}}|_{\Psi}) \oplus \Pol_{\mathbb{R}}(M_{4})E\bigg]E^{q-1},$$
where we also use what established in \eqref{Z^2(X, d_S)}. 
On the other hand, explicit computations show that 
$$\Imag(d^{2q-1}_{\widetilde{S}}|_{\Psi}) \simeq \bigg[\Big\{ Q\sum_{j<k} \epsilon_{ijk} M_{i}C_{j}C_{k}: Q \in \mathcal{O}_{X_{0}}\Big\} \oplus \Imag(d^{1}_{\widetilde{S}}|_{\Psi})\bigg] \cdot E^{q-1}.$$
From this, we immediately deduce the claimed isomorphism for the cohomology groups in degree {\small{$2q$}}. 
\end{proof}

\begin{oss}
As expected, the BRST cohomology group of degree {\small{$0$}} coincides with the space of all polynomials that are invariant under the gauge group action, that is, the classical observables of the initial gauge theory {\small{$(X_{0}, S_{0})$}}. 
\end{oss}

\section{A generalized notion of Lie algebra cohomology}
\label{Section: A generalized notion of Lie algebra cohomology}
The main objective of this section will be the introduction of a new notion of {\em generalized Lie algebra cohomology.} Thanks to this different perspective, we will reach a deeper understanding of the structure of the BRST cohomology complex constructed in Section \ref{section: Application to a $U(2)$-matrix model} for our matrix model with a {\small{$U(2)$}}-gauge symmetry (cf. Section \ref{Section: BRST and generalized Lie algebra cohomology}). We mention that there have been earlier attempts to relate BRST cohomology complex to Lie algebra cohomology complex (cf. \cite{Holten}). Lie algebra cohomology was first introduced by Chevalley and Eilenberg \cite{Chev-Eil}: for this reason it is also known as {\em Chevalley-Eilenberg cohomology}. Other classical references are \cite{Hoch} and \cite{homol}. Moreover, in recent years the classical notion of Lie algebra cohomology has been generalized to the case of superlie algebras (cf. \cite{superlie1}) and also adapted to the context of Hopf algebras \cite{Dubois-Violette}.\\
\\
Differently to what achieved in the aforementioned references, here we pursue and obtain a generalization of the classical notion of Lie algebra cohomology which would include the possibility of having generators of degree {\small{$d> 1$}}. What forces us to require this extra flexibility to the classical notion of Lie algebra cohomology is the presence of ghost fields of ghost degree {\small{$d> 1$}}, as result of the application of the BV construction. However, the presence of higher degree ghost fields is typical in the context of the BV formalism. Hence, we expect the following notion of generalized Lie algebra cohomology to be relevant also in a more general context and for many other classes of models.\\
\\
In consideration of the context where this notion will be applied, we take {\small{$\mathbb{R}$}} as ground-field. However, the whole construction  is expected to work in a more general setting. 
 
 \begin{definition}
\label{def module order p}
Given a vector space {\small{$\mathfrak{h}$}}, an {\em {\small{$\mathfrak{h}$}}-module of degree $p$}, for {\small{$p \in \mathbb{Z}$}}, is a pair {\small{$(\{V_{i}\}, \{\alpha_{j}\})$}}, where {\small{$\{V_{i}\}$}} for {\small{$i = 1, \dots, n$}} is a collection of vector spaces and 
  $$\alpha_{j}: \mathfrak{h}\longrightarrow \Lin(V_{j}, V_{j+p})$$
are linear maps, labeled by {\small{$j= 1, \dots, n-p$}} if {\small{$p\geqslant 0$}} or by {\small{$j= 1-p, \dots, n$}} if {\small{$p<0$}}. 
 \end{definition}

\begin{definition}
\label{Lie cohom def}
For {\small{$\mathfrak{h}$}} a Lie algebra and {\small{$(\{V_{i}\}, \{\alpha_{j}\})$}} a module of order {\small{$p$}} on {\small{$\mathfrak{h}$}}, let {\small{$\{ \beta_{i}: V_{i}\rightarrow V_{i+p} \}$}}, with {\small{$i=1, \dots, n-p$}} be a collection of linear maps satisfying the following conditions:
\begin{enumerate}
 
\item \vspace{1,5mm}{\small{$\alpha_{i+p}(x)\circ \beta_{i} + \epsilon \beta_{i+p} \circ \alpha_{i}(x)= 0 $}};
\item \vspace{1,5mm}  {\small{$\alpha_{i+p}(x_{1})\circ \alpha_{i}(x_{2}) + \epsilon \alpha_{i+p}(x_{2})\circ \alpha_{i}(x_{1}) = \beta_{i+p} (\alpha_{i}([x_{1}, x_{2}]_{\epsilon}))$;}}
\item \vspace{1,5mm} {\small{$\beta_{i+p} \circ \beta_{i} = 0\quad \quad $ (if $\epsilon = +1$);}}
\end{enumerate}
{\small{$\forall x, x_{1}, x_{2} \in \mathfrak{h}$}} and {\small{$\forall i$}}, where {\small{$\epsilon$}} is fixed to be either {\small{$+1$}} or {\small{$-1$}}, condition (3) is imposed only for {\small{$\epsilon = +1$}}, and {\small{$[ - , - ]_{-1}$}} denotes the Lie algebra bracket on {\small{$\mathfrak{h}$}} while the notation {\small{$[ - , - ]_{+1}$}} is used for the anticommutator. Then the {\em generalized Lie algebra cohomology} of {\small{$\mathfrak{h}$}} over {\small{$(\{V_{i}\}, \{\alpha_{j}\})$}} with parity {\small{$\epsilon$}} is defined to be a complex with cochain spaces 
$$\mathcal{C}^{j}_{\epsilon}(\mathfrak{h}, V) = \bigoplus_{i=1}^{n} \mathcal{C}_{\epsilon}^{i,j}(\mathfrak{h}, V_{i}) \quad \mbox{ where } \quad 
\mathcal{C}_{\epsilon}^{i,j}(\mathfrak{h}, V_{i}) = Sym^{j}_{\epsilon}(\mathfrak{h}, V_{i}) 
$$
and with coboundary operator {\small{$d_{\epsilon}^{j}: \mathcal{C}^{j}_{\epsilon}(\mathfrak{h}, V) \rightarrow \mathcal{C}^{j+1}_{\epsilon}(\mathfrak{h}, V)$}} defined as 
$$d^{j}_{\epsilon} = \oplus_{i=1}^{n-p} d^{j, i}_{\epsilon} \quad \quad \mbox{ for } \quad \quad  d^{j, i}_{\epsilon}: \mathcal{C}^{j}_{\epsilon}(\mathfrak{h}, V_{i}) \rightarrow \mathcal{C}^{j+1}_{\epsilon}(\mathfrak{h}, V_{i+p})$$
where, given {\small{$\varphi \in \mathcal{C}_{\epsilon}^{j, i}(\mathfrak{h}, V_{i})$}} and {\small{$x_{1}, \dots, x_{j+1} \in \mathfrak{h}$}}\vspace{-2mm}
{\small{$$
d^{j, i}_{\epsilon}(\varphi)|_{x_{1}, \dots, x_{j+1}} := \frac{1}{j+1}\Big[\sum_{r=1}^{j+1} \epsilon^{r+1} \alpha_{i}(x_{r})|_{\varphi(x_{1}, ..,\hat{x}_{r}, ..,  x_{j+1})} - \sum_{r<s} \epsilon^{r+s} \beta_{i}|_{\varphi([x_{r}, x_{s}]_{\epsilon}, .., \hat{x}_{r}, \hat{x}_{s}, ..)}\Big]. 
$$
}}
\end{definition}

In the above definition, the notation {\small{$Sym_{\epsilon}$}} is introduced to keep track of the parity of the complex. Indeed, for even parity {\small{$\epsilon = +1$}} the symmetric algebra would be even and the notation {\small{$Sym_{\epsilon}(\mathfrak{h}, V_{i})$}} would simply indicate k-linear maps on {\small{$\mathfrak{h}$}} with values in {\small{$V_{i}$}}. Contrary, for parity {\small{$\epsilon = -1$}}, {\small{$Sym_{\epsilon}(\mathfrak{h}, V_{i})$}}   would be denoting  the collection of antisymmetric k-linear maps on {\small{$\mathfrak{h}$}} with values in {\small{$V_{i}$}}.

\begin{oss}
The notion introduced in Definition \ref{Lie cohom def} can be viewed as a generalization of the classical notion of Lie algebra cochain complex, as introduced  by Chevalley and Eilenberg \cite{Chev-Eil} and later developed by Hochschild and Serre \cite{Hoch}. Indeed, the classical definition can be recovered by taking {\small{$p=0$, $n=1$}} and {\small{$\epsilon= -1$}}. Under these assumptions, we would be considering an ordinary module {\small{$V$}} over the Lie algebra {\small{$\mathfrak{h}$}} and the single map {\small{$\beta$}} would be given by the identity on {\small{$V$}}. 
\end{oss}

\begin{prop}
Given a Lie algebra {\small{$\mathfrak{h}$}}, a module {\small{$(\{V_{i}\}, \{\alpha_{j}\})$}} of order p and a collection of maps {\small{$\{ \beta_{i}\}$}} satisfying the conditions listed in Definition \ref{Lie cohom def}, the corresponding pair {\small{$(\mathcal{C}^{\bullet}_{\epsilon}(\mathfrak{h}, V), d_{\epsilon})$}} defines a cochain complex both for {\small{$\epsilon=+1$}} and {\small{$\epsilon=-1$}}.
\end{prop}

\begin{proof}
From the bracket {\small{$[ - , - ]_{\epsilon}$}} being symmetric/antisymmetric with respect to the exchange of the two entries for {\small{$\epsilon= +1$}} and {\small{$\epsilon = -1$}} respectively, it straightforwardly follows that the operator {\small{$d^{j, i}_{\epsilon}$}} preserves the type of map to which it is applied. Hence, we only have to check whether {\small{$d_{\epsilon}^{j}$}} satisfies also the coboundary condition, that is, if it holds
$$d_{\epsilon}^{j+1, i+p}\circ d_{\epsilon}^{j, i} = 0$$
for any value of the indices {\small{$i, j$}}. However, since the map {\small{$d_{\epsilon}^{j+1, i+p}$}} is defined to be the zero map for {\small{$i+p>n $}}, we could restrict ourselves to the case in which {\small{$i\leqslant n-2p$}}. By proceeding with an explicit computation, one can verify that the composition {\small{$d_{\epsilon}^{j+1, i+p}\circ d_{\epsilon}^{j, i}$}} applied to a fixed cochain {\small{$\varphi \in \mathcal{C}_{\epsilon}^{j, i}(\mathfrak{h}, V_{i})$}} and evaluated on a collection of generic elements {\small{$x_{1}, \dots, x_{j+2} \in \mathfrak{h}$}} can be viewed as sum of four terms
$$d^{j+1, i+p}_{\epsilon} \circ d^{j, i}_{\epsilon}(\varphi)|_{[x_{1}, \dots, x_{j+2}]} = \frac{1}{(j+1)(j+2)} \Big([I] + [II] + [III] + [IV]\Big)$$
where we use the following notation:
{\small{$$[I] := \sum_{r=1}^{j+2} \sum_{t<r} \epsilon^{r+t} \big[ \alpha_{i+p}(x_{t}) \circ \alpha_{i}(x_{r}) + \epsilon \alpha_{i+p}(x_{r})\circ \alpha_{i}(x_{t}) - \beta_{i+p}\circ \alpha_{i}([x_{t}, x_{r}]_{\epsilon}) \big]|_{\varphi(\dots, \hat{x}_{t}, \dots, \hat{x}_{r}, \dots)}, $$}}
which is zero due to condition {\small{$(2)$}} required in Definition \ref{Lie cohom def}, and 
{\small{\begin{multline*} 
\hspace{-2mm}[II]:= \\
- \Big[\sum_{r<s<t} \epsilon^{n+1}  (\alpha_{i+p}(x_{t})\circ \beta_{i} + \epsilon \beta_{i+p}\circ \alpha_{i}(x_{t})) + \sum_{r<t<s} \epsilon^{n}  (\alpha_{i+p}(x_{t})\circ \beta_{i} + \epsilon \beta_{i+p}\circ \alpha_{i}(x_{t})) \\
+ \sum_{t<r<s} \epsilon^{n+1} ( \alpha_{i+p}(x_{t})\circ \beta_{i} + \epsilon \beta_{i+p}\circ \alpha_{i}(x_{t}))\Big]|_{\varphi([x_{r}, x_{s}]_{\epsilon}, \dots, \hat{x}_{r}, \hat{x}_{s}, \dots)},
\end{multline*}
}}
with {\small{$n= r+s+t$}}, where each term in these summations is zero because of condition {\small{$(1)$}}. Finally, in case of a generalized Lie algebra complex of odd parity, i.e., if {\small{$\epsilon = -1$}}, the vanishing of terms {\small{$[III]$, $[IV]$}} is guaranteed by the properties satisfied by the Lie bracket {\small{$[-, -]_{-1}$}}. Indeed, the Jacobi identity and the linearity of {\small{$\varphi$}} allow to conclude that  
$$[III]:= \sum_{r<s<t} \epsilon^{r+s+t} \beta_{i+p} \circ \beta_{i}|_{\varphi(w, \dots, \hat{x}_{r}, \hat{x}_{s}, \hat{x}_{t})} = 0
$$
for 
$$w= \big[ [x_{r}, x_{s}]_{-1}, x_{t}\big]_{-1} +  \big[ [x_{s}, x_{t}]_{-1}, x_{r}\big]_{-1} +  \big[ [x_{t}, x_{r}]_{-1}, x_{s}\big]_{-1},$$ 
while the antisymmetry in the exchange of the entries of both the map {\small{$\varphi$}} and of the Lie bracket ensures that 
{\small{\begin{multline*}
[IV]:= \sum_{u<v<r<s} \epsilon^{n} (\beta_{i+p} \circ \beta_{i}|_{\varphi([x_{u}, x_{v}]_{-1}, [x_{r}, x_{s}]_{-1}, \dots)} + \beta_{i+p} \circ \beta_{i}|_{\varphi([x_{r}, x_{s}]_{-1}, [x_{u}, x_{v}]_{-1}, \dots)})\\
+ \sum_{u<r<v<s} \epsilon^{n+1} (\beta_{i+p} \circ \beta_{i}|_{\varphi([x_{u}, x_{v}]_{-1}, [x_{r}, x_{s}]_{-1}, \dots)} + \beta_{i+p} \circ \beta_{i}|_{\varphi([x_{r}, x_{s}]_{-1}, [x_{u}, x_{v}]_{-1}, \dots)})\\
+ \sum_{u<r<s<v} \epsilon^{n} (\beta_{i+p} \circ \beta_{i}|_{\varphi([x_{u}, x_{v}]_{-1}, [x_{r}, x_{s}]_{-1}, \dots)} + \beta_{i+p} \circ \beta_{i}|_{\varphi([x_{r}, x_{s}]_{-1}, [x_{u}, x_{v}]_{-1}, \dots)}) =0
\end{multline*}
}}
with {\small{$n=r+s+u+v$}}. On the other hand, from condition {\small{$(3)$}} straightforwardly follows that terms {\small{$[III], [IV]$}} are zero also in the even case, i.e. for {\small{$\epsilon = +1$}}. Hence, we conclude that the operator {\small{$d^{\bullet}_{\epsilon}$}} satisfies the coboundary condition both in the even and in the odd case and therefore the pair {\small{$(\mathcal{C}^{\bullet}_{\epsilon}(\mathfrak{h}, V), d_{\epsilon})$}} defines a cochain complex both for {\small{$\epsilon=+1$}} and {\small{$\epsilon=-1$}}.
\end{proof}

\section{BRST and generalized Lie algebra cohomology: a U(2)-model}
\label{Section: BRST and generalized Lie algebra cohomology}
The notion of generalized Lie algebra cohomology introduced in Section \ref{Section: A generalized notion of Lie algebra cohomology} is here used to rewrite the gauge-fixed BRST complex for our {\small{$U(2)$}}-matrix model of interest. By separating the generators of the BRST cohomology complex between real/bosonic on one side and Grassmannian/fer\-mio\-nic on the other, the whole BRST cohomology complex can then be described as a {\em shifted double complex} in this generalized Lie algebra cohomology setting. \\
\\
In what follows, {\small{$\mathfrak{g}$}} will denote the Lie algebra generated by the matrices {\small{$i\sigma_{1}$, $i\sigma_{2}$}} and {\small{$i\sigma_{3}$}}, for {\small{$\sigma_{1}$, $\sigma_{2}$, $\sigma_{3}$}} the Pauli matrices listed in Equation (\ref{basis}), seen as the dual of the ghost fields {\small{$C_{1}$, $C_{2}$, $C_{3}$}}. Hence, {\small{$\mathfrak{g}\cong su(2)$}} as Lie algebra. Moreover, we denote by {\small{$\mathfrak{h}$}} the Lie algebra generated by the dual of the ghost field {\small{$E$}}, which then satisfies {\small{$\mathfrak{h}\cong u(1)$}}. Finally, the notations {\small{$\mathcal{C}_{+}^{\bullet}/\mathcal{C}_{-}^{\bullet}$}} are respectively used to denote the generalized Lie algebra cochain complexes of even/odd parity.\\
\\
Let 
$$V_{i} = \mathcal{C}_{-}^{i-1}(\mathfrak{g}, \mathcal{O}_{X_{0}})$$  
for {\small{$i= 1, \dots, 4$}} be the collection of vector spaces determined by the classical Lie algebra cohomology complex of {\small{$\mathfrak{g}\cong su(2)$}} over the module {\small{$\mathcal{O}_{X_{0}}$}}, with the {\small{$\mathfrak{g}$}}-module structure {\small{$\omega: \mathfrak{g} \rightarrow Lin(\mathcal{O}_{X_{0}})$}} defined to be
$$
\omega(x)|_{f} = - \sum_{i, j, k} \epsilon_{ijk} (\partial_{i} f) M_{j}x_{k},
$$
where {\small{$\epsilon_{ijk}$}} denotes the totally antisymmetric tensor in three indices {\small{$i, j, k = 1, 2, 3$}} with {\small{$\epsilon_{123} = 1$}}, {\small{$x \in \mathfrak{g}$}} with {\small{$x = \sum_{j} x_{j} i  \sigma_{j}$}} and {\small{$f \in \mathcal{O}_{X_{0}}$}}. This collection of vectors spaces {\small{$\{V_{i} \}$}} can be completed to a module of order {\small{$p=-1$}} structure on {\small{$\mathfrak{h}$}} by defining a collection of linear maps 
$$\alpha_{i}: \mathfrak{h} \rightarrow Lin(\mathcal{C}_{-}^{i-1}(\mathfrak{g}, \mathcal{O}_{X_{0}}), \mathcal{C}_{-}^{i-2}(\mathfrak{g}, \mathcal{O}_{X_{0}}))$$
for {\small{$i=2, 3, 4$}}. Explicitly, given a generic cochain {\small{$\varphi \in \mathcal{C}^{1}_{-}(\mathfrak{g}, \mathcal{O}_{X_{0}})$}}, {\small{$\varphi = \sum_{i} f_{i}C_{i}$}}, for {\small{$i=1,2, 3$}}, we define the map {\small{$\alpha_{2}$}} to be
\begin{equation}
\label{alpha2}
\alpha_{2}(\tau)|_{\varphi} = \sum_{i=1}^{3} \frac{\partial \varphi}{\partial C_{i}} M_{i}
\end{equation}
on the generator {\small{$\tau$}} of the Lie algebra {\small{$\mathfrak{h}$}}, and then extend it to the whole {\small{$\mathfrak{h}$}} by linearity. Analogously, given a generic cochain {\small{$\varphi \in \mathcal{C}^{2}_{-}(\mathfrak{g}, \mathcal{O}_{X_{0}})$}}, with {\small{$\varphi = \sum_{i<j} f_{ij}C_{i}C_{j}$}}, we uniquely determine the linear map {\small{$\alpha_{3}$}} by requiring that
\begin{equation}
\label{alpha3}
\alpha_{3}(\tau)|_{\varphi} = \sum_{i<j} \frac{\partial \varphi}{\partial C_{i}} M_{i} C_{j} - \frac{\partial \varphi}{\partial C_{j}} M_{j} C_{i}.
\end{equation}
Equivalently, given a decomposition of {\small{$\varphi$}} of the form {\small{$\varphi = \sum \varphi_{i} \varphi_{j}$}} for {\small{$\varphi_{i}, \varphi_{j} \in \mathcal{C}^{1}_{-}(\mathfrak{g}, \mathcal{O}_{X_{0}})$}}, it holds that
$$\alpha_{3}(\tau)|_{\varphi} = \sum_{i, j} \alpha_{2}(\tau)|_{\varphi_{i}} \varphi_{j} - \varphi_{i} \alpha_{2}(\tau)|_{\varphi_{j}}.$$
At last, the map {\small{$\alpha_{4}$}} is specified by requiring that
\begin{equation}
\label{alpha4}
\alpha_{4}(\tau)|_{\varphi}= \sum_{j<k}\epsilon_{ijk}\frac{\partial \varphi}{\partial C_{i}} M_{i} C_{j}C_{k}
\quad \mbox{or} \quad  \alpha_{4}(\tau)|_{\varphi}= \sum_{j<k}\epsilon_{ijk} \alpha_{1}(\tau)|_{\varphi_{i}}\varphi_{j}\varphi_{k}
\end{equation}
for {\small{$\varphi \in \mathcal{C}^{3}_{-}(\mathfrak{g}, \mathcal{O}_{X_{0}})$}},
with decomposition {\small{$\varphi = \varphi_{1}\varphi_{2}\varphi_{3}$}}, where {\small{$\varphi_{i} \in \mathcal{C}^{1}_{-}(\mathfrak{g}, \mathcal{O}_{X_{0}})$}}.

\begin{prop}
\label{module structure}
The pair {\small{$(\{ V_{i}\}, \{ \alpha_{j}\})$}}, with vector spaces
$$V_{i} = \mathcal{C}_{-}^{i-1}(\mathfrak{g}, \mathcal{O}_{X_{0}}),$$ 
{\small{$i=1, \dots, 4$}}, and with linear maps {\small{$\{\alpha_{j}\}$}}, {\small{$j=2, 3, 4$}} defined in \eqref{alpha2}, \eqref{alpha3}, \eqref{alpha4} respectively, induces a generalized Lie algebra cohomology complex of even parity on the Lie algebra {\small{$\mathfrak{h}\cong u(1)$}}, with linear maps {\small{$\beta_{j}:V_{j} \rightarrow V_{j+p}$}} fixed to be the zero map. 
\end{prop}

\begin{proof}
By referring to Definition \ref{Lie cohom def}, the only condition that still has to be checked is the following one: 
$$\alpha_{i+1}(x_{1}) \circ \alpha_{i}(x_{2}) + \alpha_{i+1}(x_{2}) \circ \alpha_{i}(x_{1}) =0,$$
for {\small{$x_{1}, x_{2} \in \mathfrak{g}$}}. This identity can be immediately verified by a direct computation. 
\end{proof}

In the following {\small{$(\mathcal{C}^{j}_{+}(\mathfrak{h}, \mathcal{C}_{-}^{i}(\mathfrak{g}, \mathcal{O}_{X_{0}})), d_{+}^{j, i}, d_{-}^{j,i})$}} will denote the even generalized Lie algebra cohomology complex of {\small{$\mathfrak{h}$}} over the module {\small{$(\mathcal{C}^{i-1}_{-}(\mathfrak{g}, \mathcal{O}_{X_{0}}), \{ \alpha_{j}\})$}} of order {\small{$p=-1$}} just constructed, where
$$d_{+}^{i,j}: \mathcal{C}^{j}_{+}(\mathfrak{h}, \mathcal{C}_{-}^{i}(\mathfrak{g}, \mathcal{O}_{X_{0}}) \rightarrow \mathcal{C}^{j+1}_{+}(\mathfrak{h}, \mathcal{C}_{-}^{i-1}(\mathfrak{g}, \mathcal{O}_{X_{0}})$$
and 
$$d_{-}^{i,j}: \mathcal{C}^{j}_{+}(\mathfrak{h}, \mathcal{C}_{-}^{i}(\mathfrak{g}, \mathcal{O}_{X_{0}}) \rightarrow \mathcal{C}^{j}_{+}(\mathfrak{h}, \mathcal{C}_{-}^{i+1}(\mathfrak{g}, \mathcal{O}_{X_{0}})$$
for {\small{$d_{+}^{i,j}$}} the coboundary operator of the even cohomology complex on {\small{$\mathfrak{h}$}} and {\small{$d_{-}^{i,j}:= Id_{\mathfrak{h}^{i}} \oplus d^{j}_{-}$}}, with {\small{$d_{-}^{j}$}} the coboundary operator of the cohomology complex {\small{$\mathcal{C}^{j}_{-}(\mathfrak{g}, \mathcal{O}_{X_{0}})$}}.

\begin{theorem}
\label{corrispondenza cochain}
The gauge-fixed BRST cohomology complex {\small{$(\mathcal{C}^{\bullet}(\widetilde{X}|_{\Psi}, d_{\widetilde{S}}|_{\Psi}), d_{\widetilde{S}}|_{\Psi})$}} of a {\small{$U(2)$}}-matrix model can be identified, by duality, with the weighted total complex induced by the double complex {\small{$(\mathcal{C}^{j}_{+}(\mathfrak{h}, \mathcal{C}_{-}^{i}(\mathfrak{g}, \mathcal{O}_{X_{0}})), d_{+}^{j, i}, d_{-}^{j,i})$}}. Precisely, at the level of cochain spaces, the following identification holds:
$$\mathcal{C}^{k}(\widetilde{X}|_{\Psi}, d_{\widetilde{S}}|_{\Psi}) \cong \bigoplus_{2i+j= k} \mathcal{C}_{+}^{i}(\mathfrak{h}, \mathcal{C}^{j}_{-}(\mathfrak{g}, \mathcal{O}_{X_{0}}))$$
for {\small{$k, i \in \mathbb{Z}_{\geqslant 0}$, $j=0, \dots, 3$}}, and, at the level of the coboundary operator it holds that:
$$d_{\widetilde{S}}|_{\Psi} \cong \oplus_{2i+j=k} (d^{i, j}_{+} + d_{-}^{i, j}).$$
\end{theorem}

\begin{proof}
The correspondence at the level of cochains follows straightforwardly from the identification 
$$(\mathcal{C}^{j}_{+}(\mathfrak{h}, \mathcal{C}_{-}^{i}(\mathfrak{g}, \mathcal{O}_{X_{0}})) \cong Sym_{\mathcal{O}_{X_{0}}}(\langle C_{1}, C_{2}, C_{3}\rangle \oplus \langle E \rangle)$$
and the real/Grassmannian parity of the generators {\small{$E/C_{i}$}}. Finally, the weight {\small{$2$}} given to the index {\small{$i$}} takes into account that the ghost field {\small{$E$}} has ghost degree {\small{$2$}}. \\
Concerning the coboundary operators, because both {\small{$d_{\widetilde{S}}|_{\Psi}$}} and {\small{$d^{i, j}_{+} + d_{-}^{i, j}$}} acts as graded derivations on the whole space of cochains, it is enough to check they agree at the level of the generators. This correspondence can be explicitly verified. For example, one can observe that the map {\small{$\omega$}}, which gives the module structure of {\small{$\mathcal{O}_{X_{0}}$}} and which enters the expression of {\small{$d_{-}^{i,j}$}}, coincides with the action of {\small{$d_{\widetilde{S}}|_{\Psi}$}} on the generators {\small{$M_{a}$}}.
\end{proof}

\begin{oss}
The properties which characterize the ghost fields, that is, their ghost degree and their parity, find a natural translation in terms of properties of the double complex, namely in terms of the weights given to the indices corresponding to the algebras {\small{$\mathfrak{g}$}} and {\small{$\mathfrak{h}$}} and of the even/odd parity of the generalized Lie algebra cohomology considered.
\end{oss}

Hence, it is natural to conjecture the emergence of an analogous structure also for the case of an {\small{$U(n)$}}-matrix model, with {\small{$n>2$}}. In particular, we expect the BRST cohomology complex to coincide with the weighted total complex induced by a multi-complex, where the weight of the indices and the parity of the complexes are determined, respectively, by the ghost degree and the parity of the ghost fields entering the BV construction. \\
\\
The description of the gauge-fixed BRST complex in terms of this generalized Lie algebra complex allows to detect a (shifted) double complex structure, which was not visible at the level of the BRST complex, as proved in the following proposition.

\begin{prop}
\label{relaz operat cobordo}
{\small{$(\mathcal{C}_{+}^{\bullet}(\mathfrak{h}, \mathcal{C}^{\bullet}_{-}(\mathfrak{g}, \mathcal{O}_{X_{0}})), d_{+}^{\bullet} \oplus d_{-}^{\bullet})$}} has a double complex structure, that is, it satisfies the following conditions:
$$(d_{+}^{i, j})^{2}= 0, \quad (d_{-}^{i, j})^{2}= 0 \quad \mbox{and} \quad  d_{-}^{k+1,i-1} \circ d_{+}^{k, i} = - d_{+}^{k, i+1} \circ d_{-}^{k, i}$$
for {\small{$k\geq 0$}} and {\small{$i=0, \dots, 3$.}}
\end{prop}

\begin{proof}
The first two relations automatically follows from {\small{$d_{+}^{\bullet}$}} and {\small{$ d_{-}^{\bullet}$}} being coboundary operators for the even/odd Lie algebra cohomology complex. To check the last relation, we first of all recall that by definition {\small{$d_{+}^{k,0} = 0$}} for all non-negative values of {\small{$k$}} and {\small{$d_{-}^{k, i} = 0$}} for negative values of $i$ and if {\small{$i>3$}}. Then, as proved in Theorem \ref{corrispondenza cochain}, we have that:
$$
\begin{array}{l}
 d_{\widetilde{S}}|_{\Psi}^{2k} = d_{-}^{k, 0} \oplus [d_{+}^{k-1, 2} \oplus d_{-}^{k, 2}]\\
[2ex]
 d_{\widetilde{S}}|_{\Psi}^{2k+1} = [d_{+}^{k, 1}\oplus d_{-}^{k, 1}] \oplus d_{+}^{k-1, 3}.\\
\end{array}
$$
Because {\small{$d_{\widetilde{S}}|_{\Psi}^{\bullet}$}}, {\small{$d_{+}^{\bullet}$}} and {\small{$d_{-}^{\bullet}$}} are all coboundary operators, we deduce that
$$
d_{\widetilde{S}}|_{\Psi}^{2k+1} \circ d_{\widetilde{S}}|_{\Psi}^{2k} = \left[ d_{+}^{k,1} \circ d_{-}^{k,0}\right] \oplus \left[ d_{-}^{k,1} \circ d_{+}^{k-1, 2} \oplus d_{+}^{k-1,3} \circ d_{-}^{k-1, 2}\right] =0.
$$

As the first term in the previous equation takes values in {\small{$\mathcal{C}^{k+1, 0}$}} while the second in {\small{$\mathcal{C}^{k, 2}$}}, we have verified the third relation for {\small{$i = 0,2$}}. Similarly, considering the composition {\small{$d_{\widetilde{S}}|_{\Psi}^{2k+2} \circ d_{\widetilde{S}}|_{\Psi}^{2k+1}$}}, we can verify the equation for the remaining values of the index $i$, i.e. for {\small{$i = 1, 3$}}.
\end{proof}

\subsection{Relation between the cohomology groups}
\label{Relation between the cohomology groups}
Next to revealing this extra double complex structure, the description of the gauge-fixed BRST complex {\small{$(\mathcal{C}^{\bullet}(\widetilde{X}|_{\Psi}, d_{\widetilde{S}}|_{\Psi}), d_{\widetilde{S}}|_{\Psi})$}} using this generalized Lie algebra complex allows to conclude the computation of all the BRST cohomology groups. 

\begin{theorem}
\label{calcolo Lie cohom groups}
The identification of {\small{$\mathcal{C}^{\bullet}(\widetilde{X}|_{\Psi}, d_{\widetilde{S}}|_{\Psi})$}} and {\small{$\mathcal{C}^{\bullet}_{+}(\mathfrak{h}, \mathcal{C}_{-}^{\bullet}(\mathfrak{g}, \mathcal{O}_{X_{0}}))$}} at the level of complexes induces the following isomorphism at the level of the corresponding cohomology groups:

\begin{enumerate}[$\blacktriangleright$]
 \item {\small{$H^{0}(\widetilde{X}|_{\Psi}, d_{\widetilde{S}}|_{\Psi}) \cong H^{0}_{-}(\mathfrak{g}, \mathcal{O}_{X_{0}})$}};
\item {\small{$H^{1}(\widetilde{X}|_{\Psi}, d_{\widetilde{S}}|_{\Psi}) \cong H^{0}_{+}(\mathfrak{h}, H^{1}_{-}(\mathfrak{g}, \mathcal{O}_{X_{0}}))$;}}
\item {\small{$H^{2}(\widetilde{X}|_{\Psi}, d_{\widetilde{S}}|_{\Psi}) \cong H^{0}_{+}(\mathfrak{h}, Z^{2}_{-}(\mathfrak{g}, \mathcal{O}_{X_{0}})) \oplus H^{1}_{+}(\mathfrak{h}, Z^{0}_{-}(\mathfrak{g}, \mathcal{O}_{X_{0}}))$;}}
\item {\small{$H^{2k+1}(\widetilde{X}|_{\Psi}, d_{\widetilde{S}}|_{\Psi}) \cong H^{k}_{+}(\mathfrak{h}, \mathcal{C}^{1}_{-}(\mathfrak{g}, \mathcal{O}_{X_{0}}))= \left\lbrace 0 \right\rbrace$,}} for {\small{$k\geq 1$;}}
\item {\small{$H^{2k}(\widetilde{X}|_{\Psi}, d_{\widetilde{S}}|_{\Psi}) \cong H^{k}_{+}(\mathfrak{h}, Z^{0}_{-}(\mathfrak{g}, \mathcal{O}_{X_{0}}))$,}} for {\small{$k\geq2$.}}
\end{enumerate}
\end{theorem}
For brevity, in what follows we denote the cochain space 
$$\mathcal{C}^{j}_{+}(\mathfrak{h}, \mathcal{C}_{-}^{i}(\mathfrak{g}, \mathcal{O}_{X_{0}}))$$ by the symbol {\small{$\mathcal{C}^{j, i}$.}}

\begin{proof}
The claimed isomorphism in {\emph{degree 0}} follows straightforwardly from noticing that 
$$\mathcal{C}^{0, 0} \cong \mathcal{O}_{X_{0}} \quad \mbox{and} \quad d_{\widetilde{S}}|_{\Psi}^{0}=d_{-}^{0, 0}.$$
Therefore:
$$H^{0}(\widetilde{X}|_{\Psi}, d_{\widetilde{S}}|_{\Psi}) = \Ker(d_{\widetilde{S}}|_{\Psi}) = \Ker(d_{-}^{0, 0}) = H^{0}_{-}( \mathfrak{g}, \mathcal{O}_{X_{0}}).$$
Concerning the case of {\emph{degree 1}}, a generic cochain {\small{$\varphi$}} in {\small{$\mathcal{C}^{1}(\widetilde{X}|_{\Psi}, d_{\widetilde{S}}|_{\Psi})$}} can be written as {\small{$\varphi= \sum_{i=1}^{3}f_{i}C_{i},$}} for {\small{$f_{i}$}}  in {\small{$\mathcal{O}_{X_{0}}$}}. Via an explicit computation, one can check that, in order to be a cocycle with respect to the coboundary operator {\small{$d_{\widetilde{S}}|_{\Psi}$}}, {\small{$\varphi$}} has to satisfy the following conditions:
$$\sum_{i \neq a} (-M_{i} \partial_{a} + M_{a} \partial_{i}) f_{i} + f_{a} = 0 \quad \mbox{and} \quad \sum_{i=1}^{3} M_{i} f_{i} = 0,$$
with {\small{$a = 1, 2, 3$}}. However, while the first set of conditions coincides with the ones for an element {\small{$\varphi \in \mathcal{C}^{0,1}$}} to be a cocycle with respect to the coboundary operator {\small{$d_{-}^{0,1}$}}, requiring that {\small{$\sum_{i=1}^{3} M_{i} f_{i} = 0$}} is exactly imposing the condition for {\small{$\varphi \in \mathcal{C}^{0,1}$}} to be a cocycle with respect to the coboundary operator {\small{$d_{+}^{0,1}$}}. Thus:
$$\Ker( d^{1}_{\widetilde{S}}|_{\Psi}) = \Ker(d_{+}^{0,1}) \cap \Ker(d_{-}^{0, 1}). $$
Since {\small{$ d^{0}_{\widetilde{S}}|_{\Psi}$}} coincides with the operator {\small{$d_{-}^{0, 0}$}}, this allows to conclude that:
$$\begin{array}{ll}
H^{1}(\widetilde{X}|_{\Psi}, d_{\widetilde{S}}|_{\Psi}) & = \Ker(d_{+}^{0,1}) \cap \dfrac{\Ker(d_{-}^{0, 1})}{\Imag(d_{-}^{0, 0})} = \Ker(d_{+}^{0,1}) \cap H^{1}_{-}(\mathfrak{g}, \mathcal{O}_{X_{0}}) \\
[3ex]
& \cong H^{0}_{+}(\mathfrak{h}, H^{1}_{-}(\mathfrak{g}, \mathcal{O}_{X_{0}})). \\
\end{array}
$$

Now we consider the case of \emph{degree 2}. In Theorem \ref{teorema con calcolo gruppi BRST per U(2)}, we proved  that 
$$ H^{2}(\widetilde{X}|_{\Psi}, d_{\widetilde{S}}|_{\Psi}) = K \oplus \Pol_{\mathbb{R}}(M_{4})E,$$ 
with 
$$K:= \Big\{ Q \sum_{i;j<k} \epsilon_{ijk} M_{i} C_{j}C_{k}: Q \in \mathcal{O}_{X_{0}}\Big\}.$$
To verify the relation between the cohomology groups claimed in this theorem, we are going to show that:

$$H^{0}_{+}(\mathfrak{h}, Z^{2}_{-}(\mathfrak{g}, \mathcal{O}_{X_{0}})) \cong K \quad \mbox{and} \quad H^{1}_{+}(\mathfrak{h}, Z^{0}_{-}(\mathfrak{g}, \mathcal{O}_{X_{0}})) \cong \Pol_{\mathbb{R}}(M_{4}) E.$$

The first isomorphism immediately follows from an explicit description of the form that a generic cochain {\small{$\varphi \in \mathcal{C}^{0,2}$}} has to have in order to belong to the intersection {\small{$\Ker(d_{+}^{0, 2}) \cap \Ker(d_{-}^{0, 2})$}}. On the other hand, the second isomorphism is implied by the fact that {\small{$\Ker(d_{+}^{1, 0})$}} coincides with {\small{$\mathcal{C}^{1,0}$}}, due to the map {\small{$d_{+}^{1,0}$}} being the zero map, and by the following identities: 
$$\Ker(d_{-}^{1, 0})\cong \Ker(d_{-}^{0, 0})\cdot E \quad \mbox{and} \quad \Imag(d_{+}^{0,1}) = \{\big(\sum_{i=1}^{3} M_{i}f_{i}\big)E, \mbox{  }f_{i} \in \mathcal{O}_{X_{0}}\}.$$
Therefore, by recalling the explicit form of {\small{$\Ker(d_{-}^{0, 0})= \Ker(d_{\widetilde{S}}|_{\Psi})$}} determined in Theorem \ref{teorema con calcolo gruppi BRST per U(2)}, one concludes that
$$H^{1}_{+}(\mathfrak{h}, Z^{0}_{-}(\mathfrak{g}, \mathcal{O}_{X_{0}}))= \dfrac{\Ker(d_{+}^{1,0})\cap \Ker(d_{-}^{1, 0})}{\Imag(d_{+}^{0,1})} = \Pol_{\mathbb{R}}(M_{4})E
$$
as claimed in the statement. Going to the case of \emph{degree  {\small{$2k+1$}}}, in Theorem \ref{teorema con calcolo gruppi BRST per U(2)} we proved that {\small{$H^{2k+1}(\widetilde{X}|_{\Psi}, d_{\widetilde{S}}|_{\Psi})$}}
is trivial for each {\small{$k\geqslant 1$}}. Therefore, to prove our statement it is enough to verify that:
$$H_{+}^{k}(\mathfrak{h}, \mathcal{C}_{-}^{1}(\mathfrak{g}, \mathcal{O}_{X_{0}})) = \left\lbrace 0 \right\rbrace \quad \quad \mbox{and}\quad \quad H_{+}^{k-1}(\mathfrak{h}, \mathcal{C}_{-}^{3}(\mathfrak{g}, \mathcal{O}_{X_{0}})) = \left\lbrace 0 \right\rbrace.$$
To check the first identity, we start by noticing that a cochain {\small{$\varphi \in \mathcal{C}^{k, 1}$}}, in order to be a cocycle in {\small{$\Ker(d_{+}^{k,1})$}}, has to be of the form:
$$\varphi = \sum_{i, j, l=1}^{3}\big( \epsilon_{ijl} M_{j}g_{l})\cdot C_{i}E^{k}$$
for {\small{$g_{l}$}} in {\small{$\mathcal{O}_{X_{0}}$}}. On the other hand, one can verify that the cochain {\small{$\psi$}}, defined by
$$\psi= \sum_{i, j, l=1}^{3}\big( \epsilon_{ijl} g_{i} C_{j} C_{l}\big) \cdot E^{k}$$ 
satisfies {\small{$d_{+}^{k, 1} (\psi)= \varphi$}}. Hence we conclude that
{\small{$H_{+}^{k}(\mathfrak{h}, \mathcal{C}_{-}^{1}(\mathfrak{g}, \mathcal{O}_{X_{0}}))$}} is trivial, for any {\small{$k\geq 1$}}. As to the second identity to prove, a direct computation shows that the map {\small{$d_{+}^{k-1, 3}$}} is injective for all {\small{$k\geq 1$}}. This concludes the proof for the case of degree {\small{$2k+1$}}. Finally, we still have to consider the case of {\emph{degree 2k}}. Accordingly to what determined in Theorem \ref{teorema con calcolo gruppi BRST per U(2)}, to complete the proof of the theorem is enough to verify that: 
$$H^{k}_{+}(\mathfrak{h}, Z^{0}_{-}(\mathfrak{g}, \mathcal{O}_{X_{0}})) = \Pol_{\mathbb{R}}(M_{4}) E^{k} \quad \mbox{and} \quad H^{k-1}_{+}(\mathfrak{h}, \mathcal{C}^{2}_{-}(\mathfrak{g}, \mathcal{O}_{X_{0}}))=0.$$
This first identity follows from the map {\small{$d_{+}^{k, 0}$}} being the zero map and from the fact that 
$$\Ker(d_{-}^{k, 0}) = \Big\{ \Big[g_{0}(M_{4}) + \sum_{j=1}^{s}(M_{1}^{2}+ M_{2}^{2}+ M_{3}^{2})^{j}g_{j}(M_{4})\Big]E^{k}: g_{i}\in \Pol_{\mathbb{R}}(M_{4})  \Big\}$$
and
$$ \Imag(d_{+}^{k-1, 1}) = \Big\{ \Big[ \sum_{i=1}^{3} M_{i}f_{i}\Big]\cdot E^{k}:  f_{i}\in \mathcal{O}_{X_{0}}\Big\},
$$
from which we deduce 
$$H^{k}_{+}(\mathfrak{h}, Z^{0}_{-}(\mathfrak{g}, \mathcal{O}_{X_{0}})) = \frac{\Ker(d_{-}^{k, 0})}{\Imag(d_{+}^{k-1, 1})} =  \Pol_{\mathbb{R}}(M_{4}) \cdot E^{k}.$$

Finally, one can check that 
$$\Ker(d_{+}^{k-1, 2}) = \Imag(d_{+}^{k-2, 3}) = \Big\{ P\Big[ \sum_{i, j, l} \epsilon_{ijl} M_{i}C_{j}C_{l}\Big], P \in \mathcal{O}_{X_{0}}\Big\} \cdot E^{k-1},$$
from which the triviality of the group {\small{$H^{k-1}(\mathfrak{h}, \mathcal{C}^{2}_{-}(\mathfrak{g}, \mathcal{O}_{X_{0}}))$}} is immediately implied.
\end{proof}

All the properties revealed by the computations carried on for the proof of the above theorem can be summarized in the following corollary. 

\begin{cor}
\label{cor suc esatta}
Let {\small{$(\mathcal{C}^{i,j}, d_{+}^{i,j}, d_{-}^{i,j})$}}, with {\small{$j\geq 0$}} and {\small{$i= 0, \dots, 3$}}, be the cochain complex analyzed in Theorem \ref{calcolo Lie cohom groups}. Then the following sequences are exact:  
\begin{equation}
\label{succ esatta}
\begin{array}{l}
0 \rightarrow \mathcal{C}^{k-2, 3} \xrightarrow{d_{+}^{k-2, 3}} \mathcal{C}^{k-1, 2} \xrightarrow{d_{+}^{k-1, 2}} \mathcal{C}^{k, 1}\xrightarrow{d_{+}^{k, 1}} \widetilde{\mathcal{C}}^{k+1, 0}\rightarrow 0,
\end{array}
\end{equation}
and 
$$0 \rightarrow \mathcal{C}_{+}^{k-2}(\mathfrak{h}, Z^{3}_{-}) \xrightarrow{d_{+}^{k-2, 3}} \mathcal{C}_{+}^{k-1}(\mathfrak{h}, Z^{2}_{-})\xrightarrow{d_{+}^{k-1, 2}} \mathcal{C}_{+}^{k}(\mathfrak{h}, Z^{1}_{-}) \xrightarrow{d_{+}^{k, 1}} \widetilde{\mathcal{C}}^{k+1}_{+}(\mathfrak{h}, \widetilde{W})\rightarrow 0,
$$
with \\
\\
\phantom{m} \quad  {\small{$\widetilde{\mathcal{C}}^{k+1, 0} := \mathcal{C}_{+}^{k+1}(\mathfrak{h}, \mathcal{C}^{0}_{-}(\mathfrak{g}, W))$}}, 
for {\small{$W:= \mathcal{O}_{X_{0}} \backslash \Pol_{\mathbb{R}}(M_{4})$}},  \\
\phantom{m} \quad  {\small{$\mathcal{C}_{+}^{j}(\mathfrak{h}, Z^{i}_{-}):= \mathcal{C}_{+}^{j}(\mathfrak{h}, Z^{i}_{-}(\mathfrak{g}, \mathcal{O}_{X_{0}}))$}} and
{\small{$\widetilde{\mathcal{C}}^{k+1}_{+}(\mathfrak{h}, \widetilde{W}) := \mathcal{C}_{+}^{k+1}(\mathfrak{h}, \mathcal{C}_{-}^{0}(\mathfrak{g}, \widetilde{W}))$}} for 
$$\widetilde{W}= \bigg\{ \sum_{j = 1}^{r} (M_{1}^{2}+ M_{2}^{2} + M_{3}^{2})^{j} g_{j}(M_{4}):  g_{j}(M_{4})\in \Pol_{\mathbb{R}}(M_{4}), \quad r \in \mathbb{N} \bigg\},$$
i.e. it holds  {\small{$Z^{0}_{-}(\mathfrak{g}, \mathcal{O}_{X_{0}}) = \Pol_{\mathbb{R}}(M_{4}) \oplus \widetilde{W}$}}
\end{cor}

\begin{proof}
The exactness of the first sequence follows from the coboundary operator {\small{$d_{+}^{k-2, 3}$}} being injective, the triviality of the cohomology groups {\small{$H^{k-1}_{+}(\mathfrak{h}, \mathcal{C}^{2}_{-}(\mathfrak{g}, \mathcal{O}_{X_{0}}))$}} and {\small{$H^{k}_{+}(\mathfrak{h}, \mathcal{C}^{1}_{-}(\mathfrak{g}, \mathcal{O}_{X_{0}}))$}} and, finally, from the fact that
$$
 \Imag(d_{+}^{k, 1}) \cong [\mathcal{O}_{X_{0}} \backslash \Pol_{\mathbb{R}}(M_{4})]E^{k+1}= \mathcal{C}_{+}^{k+1}(\mathfrak{h}, \mathcal{C}_{-}^{0}(\mathfrak{g}, W))
 $$
which implies the subjectivity of the map {\small{$d_{+}^{k, 1}$}} and hence allows to conclude that the first sequence is exact. On the other hand, because the second sequence is obtained from the first via the restriction of the domains and codomains of the coboundary operators {\small{$d_{+}^{\bullet, \bullet}$}} to cochain spaces with coefficients in the cocycle spaces {\small{$Z^{\bullet}_{-}(\mathfrak{g}, \mathcal{O}_{X_{0}})$}}, its exactness follows from the exactness of (\ref{succ esatta}), the only thing to check being that this second sequence is actually well defined. More explicitly, we have to verify that, given an element {\small{$\varphi$}} in {\small{$\mathcal{C}_{+}^{i}(\mathfrak{h}, Z^{j}_{-})$}}, with {\small{$i\geq 0$}} and {\small{$j= 1, 2, 3$}}, then {\small{$d_{+}^{i,j}(\varphi)$}} is an element in {\small{$\mathcal{C}_{+}^{i+1}(\mathfrak{h}, Z^{j-1}_{-})$}}. \\
\\
By hypothesis, {\small{$\varphi$}} is given by a product {\small{$\varphi = \psi \cdot E^{i}$}} with {\small{$\psi \in Z^{j}_{-}(\mathfrak{g}, \mathcal{O}_{X_{0}})$}}. We show that {\small{$d_{+}^{i,j}(\varphi) = \chi \cdot E^{i}$}} for some {\small{$\chi$}} in {\small{$Z^{j-1}_{-}(\mathfrak{g}, \mathcal{O}_{X_{0}})\cdot E.$}}
Let us define {\small{$\chi:=d_{+}^{0,j}(\psi)$}}. Then we only have to show that {\small{$\chi$}} is an element of {\small{$\Ker(d_{-}^{1, j-1})$}}.
Recall what we proved in Proposition \ref{relaz operat cobordo}, which implies
$$d_{-}^{1, j-1}(\chi)= d_{-}^{1, j-1}(d_{+}^{0,j}(\psi))= - d_{+}^{0, j+1}(d_{-}^{0,j}(\psi))=0,$$
where we are using the fact that {\small{$\psi$}} belongs to {\small{$Z^{j}_{-}(\mathfrak{g}, \mathcal{O}_{X_{0}})$}}. Therefore, the second sequence is well defined and hence exact. 
\end{proof}

Finally, we can prove the following lemma about the relation between cocycles in the gauge-fixed BRST complex and the different cocycle/coboundary groups on the side of the corresponding generalized Lie algebra complex. 

\begin{lemma}
\label{lemma propr}
Let {\small{$d^{\bullet}_{\widetilde{S}}|_{\Psi}$, $d_{+}^{\bullet}$, $d_{-}^{\bullet}$}} and {\small{$\left\lbrace \mathcal{C}^{k, i}\right\rbrace$, $k \in \mathbb{N}_{0}$}}, {\small{$i=0, \dots, 3$}}, be the same coboundary operators and the same collection of cochain spaces as in Theorem \ref{calcolo Lie cohom groups}. Then, for each {\small{$k\geq1$}}, it holds
 $$\Ker(d_{\widetilde{S}}^{2k}|_{\Psi}) = \Ker(d_{-}^{k, 0}) \oplus \left[ \Imag( d_{-}^{k-1, 1}) + \Ker(d_{+}^{k-1, 2})\right].$$
\end{lemma}

\begin{proof}
Let us starting by noticing that a generic cochain {\small{$\varphi$}} of ghost degree {\small{$2k$}} can be written as sum of two cochains {\small{$\varphi_{k, 0}$}} and {\small{$\varphi_{k-1, 2}$}}, which belong to the cochain spaces {\small{$\mathcal{C}^{k, 0}$}} and {\small{$\mathcal{C}^{k-1, 2}$}} respectively. Thus
$$d^{2k}_{\widetilde{S}}|_{\Psi}(\varphi)= (d_{+}^{k-1, 2} + d_{-}^{k-1, 2}) (\varphi_{k-1, 2})  + d_{-}^{k, 0}(\varphi_{k, 0}).$$
Because {\small{$d_{-}^{k, 0}(\varphi_{k, 0})$}} and {\small{$d_{+}^{k-1, 2}(\varphi_{k-1, 2})$}} are elements in {\small{$\mathcal{C}^{k,1}$}} while {\small{$d_{+}^{k-1, 2}(\varphi_{k-1, 2})$}} belongs to {\small{$\mathcal{C}^{k-1, 3}$}}, to have that {\small{$\varphi$}} is a cocycle element is equivalent to imposing
$$ d_{-}^{k, 0}(\varphi_{k, 0}) = - d_{+}^{k-1, 2}(\varphi_{k-1, 2}) \quad \mbox{and} \quad d_{-}^{k-1, 2}(\varphi_{k-1, 2}) = 0.$$
Considering the intersection of {\small{$\Imag(d_{-}^{k, 0})$}} together with {\small{$\Imag(d_{+}^{k-1, 2})$}}, we see that a generic coboundary element {\small{$\alpha$}} in {\small{$\Imag(d_{-}^{k, 0})$}} on one side and a generic element {\small{$\beta$}} in {\small{$\Imag(d_{+}^{k-1, 2})$}} on the other take the following form, respectively:
$$
\alpha =  - \Big[ \sum_{i, j , l}\epsilon_{ijl} \ M_{i}(\partial_{j}f) C_{l}\Big] \cdot E^{k},  \quad  \beta = - \Big[ \sum_{i, j , l}\epsilon_{ijl} \ g_{i}M_{j}C_{l}\Big] \cdot E^{k}.
$$
for some {\small{$f$, $g_{j}$}} in {\small{$\mathcal{O}_{X_{0}}$}}, Hence, to have a cochain that is a coboundary with respect to both maps {\small{$d_{+}^{k-1, 2}$}} and {\small{$d_{-}^{k, 0}$}}, the polynomials {\small{$f$}} and {\small{$g_{j}$}} need to satisfy
$$\begin{array}{lcr}
 g_{1}= M_{3}Q + \partial_{1}f \, \quad & \quad  g_{2}= - M_{2}Q - \partial_{2}f \,\quad &\quad g_{3}= M_{1}Q + \partial_{3}f. 
  \end{array}
$$
with {\small{$Q$}} a polynomial in {\small{$\mathcal{O}_{X_{0}}$.}} Therefore, for {\small{$d_{+}^{k-1, 2}(\varphi_{k-1, 2})$}} to coincide with {\small{$-d_{-}^{k. 0}(\varphi_{k, 0})$}}, we need
$$ \varphi_{k-1, 2} = \left[ \left(M_{1}Q + \partial_{3}f \right)C_{1}C_{2} - \left(M_{2}Q + \partial_{2}f\right)C_{1}C_{3} +  \left(M_{3}Q + \partial_{1}f\right)C_{2}C_{3}\right] E^{k-1} \mbox{ },$$
where {\small{$Q$}} and {\small{$f$}} are generic polynomials in {\small{$\mathcal{O}_{X_{0}}$}}. Let {\small{$A_{i}$}}, for {\small{$i = 1, 2, 3$}} be polynomials in {\small{$\mathcal{O}_{X_{0}}$}} and {\small{$A_{0}$}} in {\small{$\Pol_{\mathbb{R}}(M_{4})$}}, such that {\small{$f$}} can be rewritten as {\small{$f = - \sum_{i=1}^{3} M_{i}A_{i} + A_{0}.$}}
By defining {\small{$\widetilde{Q} := - \sum_{i = 1}^{3} \partial A_{i}/\partial M_{i},$}} {\small{$\varphi_{k-1, 2}$}} can be written as a sum {\small{$\varphi_{k-1, 2} = \lambda + \mu$}}, with {\small{$\lambda$}} an element in {\small{$\Imag(d_{-}^{k-1, 1})$}} explicitly given by 
$$\lambda = [(M_{1}\widetilde{Q} + \partial_ {3}f)C_{1}C_{2} - ( M_{2}\widetilde{Q} + \partial_ {2}f)C_{1}C_{3}+ (M_{3}\widetilde{Q} + \partial_ {1}f)C_{2}C_{3}] E^{k-1}$$
and {\small{$\mu$}} an element in {\small{$\Ker(d_{+}^{k-1, 2})$}} with
$$\mu= \left( Q - \widetilde{Q} \right) \Big[ \sum_{i, j<l} \epsilon_{ijl} M_{i}C_{j}C_{l}\Big]E^{k-1}.$$
Hence, a generic cochain {\small{$\varphi$}} in {\small{$\Ker(d_{\widetilde{S}}^{2k}|_{\Psi})$}} can be uniquely decomposed as 
$$\varphi= \varphi_{k, 0} + \lambda + \mu$$
with {\small{$\lambda$}} in {\small{$\Imag(d_{-}^{k-1, 1})$}} and {\small{$\mu$}} in {\small{$\Ker(d_{+}^{k-1, 2})$}}. Moreover, because 
$$d_{-}^{k, 0} (\varphi_{k, 0})= - d_{+}^{k-1, 2}(\varphi_{k-1, 2}) =  - d_{+}^{k-1, 2}(d_{-}^{k-1,1} (\alpha)) = d_{-}^{k, 0}(d_{+}^{k-1, 1}(\alpha)) = 0,$$
we also conclude that {\small{$\varphi_{k, 0}$}} belongs to {\small{$\Ker(d_{-}^{k, 0})$}}. This last observation implies the statement.
\end{proof}

\subsection{The shifted double complex}
In this concluding paragraph we reverse what found in the previous section by proving that a double complex satisfying the properties determined for {\small{$(\mathcal{C}^{i,j}, d_{+}^{i,j}, d_{-}^{i,j})$}} presents, at the level of cohomology groups, the same relation with the corresponding total complex that we found between gauge-fixed BRST cohomology and generalized Lie algebra cohomology. In other words, we have explicitly identified the complete list of properties which enforce the isomorphisms in Theorem \ref{calcolo Lie cohom groups} at the level of cohomology groups.

\begin{theorem}
\label{theorem double complex}
Let {\small{$(\mathcal{C}^{i, j}, d_{+}^{i, j} \oplus d_{-}^{i, j})$}} be a (shifted) double complex in the generalized Lie algebra cohomology with 
$$\mathcal{C}^{i,j}:= \mathcal{C}_{+}^{i}(\mathfrak{h}, \mathcal{C}_{-}^{j}(\mathfrak{g}, \Omega))$$
for {\small{$\Omega$}} a module over {\small{$\mathfrak{g}$}}, {\small{$\mathcal{C}_{-}^{j}(\mathfrak{g}, \Omega)$}} a module of order {\small{$p = -1$}} on {\small{$\mathfrak{h}$}} and coboundary operators  
$$d_{+}^{i,j}: \mathcal{C}^{i, j} \rightarrow \mathcal{C}^{i+1, j-1} \quad \quad d_{-}^{i,j}: \mathcal{C}^{i, j} \rightarrow \mathcal{C}^{i, j+1}.$$
Suppose that this complex satisfies the following list of properties:
\begin{enumerate}[(1)]
\item \label{propr bicomplesso1} {\small{$d_{-}^{\bullet, \bullet}$}} and {\small{$d_{+}^{\bullet, \bullet}$}} are coboundary operators, i.e., 
$$d_{-}^{k, i+1} \circ d_{-}^{k, i}= 0 \quad \mbox{and} \quad d_{+}^{k+1, i} \circ d_{+}^{k, i+1}= 0,$$
for all {\small{$k\geq 0$}} and {\small{$i= 0, 1, 2$;}}
\item \label{propr bicomplesso3} The composition for the two operators {\small{$d_{-}^{\bullet}$}} and {\small{$d_{+}^{\bullet}$}} satisfies the following relation, for all {\small{$k\geq0$}} and {\small{$i=0, \dots, 3$:}}
$$d_{-}^{k+1,i-1} \circ d_{+}^{k, i} = - d_{+}^{k, i+1} \circ d_{-}^{k, i}.$$
Moreover, {\small{$d_{+}^{k,0} = 0$}} and {\small{$d_{-}^{k, i}= 0$}}, for all {\small{$k\geq 0$, $i<0$}} and {\small{$i>3$;}}
\item \label{propr bicomplesso4} The operator {\small{$d_{+}^{k, 3}$}} is injective for all {\small{$k\geq 0$;}}
\item \label{propr bicomplesso5} The following identity holds:
 $$\Ker(d_{tot}^{2k}) = \Ker(d_{-}^{k, 0}) \oplus \left[ \Imag( d_{-}^{k-1, 1}) + \Ker(d_{+}^{k-1, 2})\right],$$ 
 for all {\small{$k\geq1$}};
\item \label{propr succ esatta} The following sequence is exact, for all {\small{$k\geq 0$:}}
$$\mathcal{C}^{k, 3} \xrightarrow{d_{+}^{k, 3}} \mathcal{C}^{k+1, 2} \xrightarrow{d_{+}^{k+1, 2}} \mathcal{C}^{k+2, 1} \xrightarrow{d_{+}^{k+2, 1}} \mathcal{C}^{k+3, 0}.
$$
\end{enumerate}
Then, the induced total complex with 
$$\mathcal{C}_{tot}^{k} := \bigoplus_{2i+j = k} \mathcal{C}^{i, j} \quad \mbox{and} \quad d_{tot}^{k} :=  \bigoplus_{2i+j = k} d_{+}^{i, j} + d_{-}^{i, j},$$
has corresponding cohomology groups satisfying the following relations:\\
\begin{enumerate}[$\blacktriangleright$]
 \item {\small{$H^{2k}((\mathfrak{g}_{1} \oplus\mathfrak{h}_{2})^{\vee}, d_{tot}) \cong H_{+}^{k}( \mathfrak{h}, Z^{0}_{-}(\mathfrak{g}, \Omega)) \oplus H_{+}^{k-1}( \mathfrak{h}, Z^{2}_{-}(\mathfrak{g}, \Omega))$; }} 
 \item {\small{$H^{2k+1}((\mathfrak{g}_{1} \oplus\mathfrak{h}_{2})^{\vee}, d_{tot}) \cong H_{+}^{k}( \mathfrak{h}, \mathcal{C}^{1}_{-}(\mathfrak{g}, \Omega)) = 0,$}}
  \end{enumerate}
for {\small{$k \geq 1$}}. In particular, for {\small{$k=0, 1$}} we have
\begin{enumerate}[$\blacktriangleright$]
 \item {\small{$H^{0}((\mathfrak{g}_{1} \oplus\mathfrak{h}_{2})^{\vee}, d_{tot}) \cong H_{-}^{0}(\mathfrak{g}, \Omega);$}}
 \item {\small{$H^{1}((\mathfrak{g}_{1} \oplus\mathfrak{h}_{2})^{\vee}, d_{tot}) \cong H_{+}^{0}( \mathfrak{h}, H^{1}_{-}(\mathfrak{g}, \Omega)).$}}
\end{enumerate}
\end{theorem}

The notation {\small{$\mathfrak{g}_{1} \oplus\mathfrak{h}_{2}$}} emphasizes that the generators in {\small{$\mathfrak{g}$}} have degree {\small{$1$}} and odd parity while {\small{$\mathfrak{h}_{2}$}} indicates that its generator has degree {\small{$2$}} and even parity.

\begin{proof}
First of all, we notice that the induced total complex {\small{$(\mathcal{C}_{tot}^{\bullet}, d_{tot}^{\bullet})$}} is indeed a cohomology complex, where {\small{$d_{tot}^{\bullet}$}} defining a coboundary operator follows from properties (\ref{propr bicomplesso1}) and (\ref{propr bicomplesso3}). Moreover, from property (\ref{propr succ esatta}), we deduce that
$$\begin{array}{lr}
   \Ker(d_{+}^{k+1, 2}) = \Imag(d_{+}^{k, 3})\mbox{ }, \quad  &  \quad  \Ker(d_{+}^{k+1, 1}) = \Imag(d_{+}^{k, 2})\mbox{ }, \quad \quad \forall k\geq 0. 
  \end{array}
$$
We start proving the theorem for the case of {\emph{total degree 2k}}. Property (\ref{propr bicomplesso5}) implies that 
$$\Ker(d_{tot}^{2k}) = \Ker(d_{-}^{k, 0}) \oplus \left[ \Imag( d_{-}^{k-1, 1}) + \Ker(d_{+}^{k-1, 2})\right].$$
However, for {\small{$k=0$}} the above equality reduces to
$$\Ker(d_{tot}^{0})  = \Ker(d_{-}^{0, 0}).$$
For what concerns the image, we have that it holds 
$$\Imag(d_{tot}^{2k-1}) = \Imag(d_{+}^{k-1, 1}) \oplus [\Imag (d_{-}^{k-1, 1}) + \Imag(d_{+}^{k-2, 3})].$$
In particular, for {\small{$k=1$}} the coboundary operator {\small{$d_{+}^{k-2, 3}$}} does not appear and so 
$$\Imag(d_{tot}^{1}) = \Imag(d_{+}^{0, 1}) \oplus \Imag (d_{-}^{0, 1}).$$
From the above equalities, it follows straightforwardly that:
$$\begin{array}{ll}
H^{2k}((\mathfrak{g}_{1} \oplus\mathfrak{h}_{2}))^{\vee}, d_{tot})   & \cong  \dfrac{\Ker(d_{-}^{k, 0}) \cap\Ker(d_{+}^{k, 0 })}{\Imag(d_{+}^{k-1, 1})} \oplus \dfrac{\Ker(d_{+}^{k-1, 2}) \cap \Ker(d_{-}^{k-1, 2})}{\Imag(d_{+}^{k-2, 3})}\\
[3ex]& =  H_{+}^{k}( \mathfrak{h}, Z^{0}_{-}(\mathfrak{g}, \Omega)) \oplus H_{+}^{k-1}( \mathfrak{h}, Z^{2}_{-}(\mathfrak{g}, \Omega)).
  \end{array}
$$
We also recall that the coboundary operator {\small{$d_{+}^{k, 0}$}} is the zero map, i.e.  {\small{$\Ker(d_{-}^{k, 0})$}} coincides with {\small{$\Ker(d_{-}^{k, 0}) \cap \Ker(d_{+}^{k, 0 })$}}, and that {\small{$\Ker(d_{+}^{k-1, 2})$}} is a subset of {\small{$\Ker(d_{-}^{k-1, 2})$}} for property (\ref{propr bicomplesso3}) and (\ref{propr succ esatta}). Finally, for {\small{$k=0$}}, we have
$$H^{0}((\mathfrak{g}_{1} \oplus\mathfrak{h}_{2})^{\vee}, d_{tot})  \cong \Ker(d_{-}^{0, 0}) = H_{-}^{0}(\mathfrak{g}, \Omega).$$
To conclude the proof of the theorem, let us consider the case of {\emph{total degree}} {\small{$2k+1$}}, for {\small{$k\geq1$}}. Given a generic cochain {\small{$\varphi$}} in {\small{$\mathcal{C}^{2k+1}_{tot}$}}, it can  be uniquely decomposed as the sum:
$$\varphi = \varphi_{k,1} + \varphi_{k-1, 3}$$
for {\small{$\varphi_{k,1}$}} in {\small{$\mathcal{C}^{k,1}$}} and {\small{$\varphi_{k-1,3}$}} in {\small{$\mathcal{C}^{k-1,3}$}}. Then, imposing that {\small{$\varphi$}} is a cocycle in the total complex is equivalent to ask that 
$$d_{-}^{k, 1}(\varphi_{k, 1}) = - d_{+}^{k-1, 3}(\varphi_{k-1, 3}) \quad \mbox{and} \quad d_{+}^{k, 1}(\varphi_{k, 1}) = 0.$$
By applying Property (\ref{propr succ esatta}), we have that {\small{$\varphi_{k,1} \in \Imag(d_{+}^{k-1, 2})$}}, i.e., that there exists an element {\small{$\alpha$}} in {\small{$\mathcal{C}^{k-1, 2}$}} such that {\small{$d_{+}^{k-1, 2}(\alpha) = \varphi_{k,1}$}}. Hence, for the relation imposed by Property (\ref{propr bicomplesso3}) among the coboundary operators, we deduce that
$$d_{+}^{k-1, 3}(\varphi_{k-1, 3}) = - d_{-}^{k, 1}( d_{+}^{k-1, 2}(\alpha)) = d_{+}^{k-1, 3}(d_{-}^{k-1, 2}(\alpha)). $$ 
Finally, using the injectivity of {\small{$d_{+}^{k, 3}$}} for all non-negative value of $k$, i.e. Property (\ref{propr bicomplesso4}), we enforce that {\small{$\varphi_{k-1, 3}$}} is an element of {\small{$\Imag(d_{-}^{k-1, 2})$}}. Thus
$$\Ker(d_{tot}^{2k+1}) = \Imag(d_{+}^{k-1, 2}) \oplus \Imag(d_{-}^{k-1, 2}).$$
In case of {\small{$k=0$}}, the above expression simplifies to
$$\Ker(d_{tot}^{1}) = \Ker(d_{+}^{0, 1}) \cap \Ker(d_{-}^{0, 1}).$$
Last, we consider {\small{$\Imag(d_{tot}^{2k})$}}. The coboundary operator {\small{$d_{tot}^{2k}$}}, with {\small{$k \geq 1$}}, can be written as a sum of coboundary operators {\small{$d_{+}^{\bullet}$}} and {\small{$d_{-}^{\bullet}$}}: 
$$d_{tot}^{2k}= d_{-}^{k, 0} \oplus [d_{+}^{k-1, 2} \oplus d_{-}^{k-1, 2}].$$
Both operators {\small{$d_{+}^{k-1, 2}$}} and {\small{$d_{-}^{k, 0}$}} take values in {\small{$\mathcal{C}^{k, 1}$}} while the image of the map {\small{$d_{-}^{k-1, 2}$}} is in {\small{$\mathcal{C}^{k-1, 3}$}}. Moreover, using Property (\ref{propr bicomplesso3}) for {\small{$i=0$}} as well as Property (\ref{propr succ esatta}), we deduce that
$$\Imag(d_{-}^{k, 0}) \subseteq \Ker(d_{+}^{k, 1}) = \Imag(d_{+}^{k-1, 2}).$$ 
Therefore,
$$\Imag(d_{tot}^{2k})= \Imag(d_{+}^{k-1, 2}) \oplus \Imag(d_{-}^{k-1, 2}).$$
In the case {\small{$k=0$}}, we immediately see that {\small{$\Imag(d_{tot}^{0})$}} coincides with {\small{$\Imag(d_{-}^{0, 0})$}}.\\
\\
To conclude, a cohomology group {\small{$H^{2k+1}((\mathfrak{g} \oplus\mathfrak{h})^{\vee}, d_{tot})$}}, for {\small{$k \geq 1$}}, satisfies the following identity:
$$H^{2k+1}((\mathfrak{g}_{1} \oplus\mathfrak{h}_{2})^{\vee}, d_{tot}) = \dfrac{\Ker(d_{+}^{k, 1})\oplus \Imag(d_{-}^{k-1, 2})}{ \Imag(d_{+}^{k-1, 2}) \oplus \Imag(d_{-}^{k-1, 2})} =  \dfrac{\Imag(d_{+}^{k-1, 2})}{\Imag(d_{+}^{k-1, 2}) } = \left\lbrace 0 \right\rbrace.$$
Finally, if {\small{$k=0$}}, we have
$$H^{1}((\mathfrak{g}_{1} \oplus\mathfrak{h}_{2})^{\vee}, d_{tot})  \cong \dfrac{\Ker(d_{+}^{0, 1}) \cap \Ker(d_{-}^{0, 1})}{\Imag(d_{-}^{0, 0})} = H_{+}^{0}( \mathfrak{h}, H^{1}_{-}(\mathfrak{g}, \Omega)).$$
\end{proof}

\section{Conclusions and outlooks}
\label{Sect: Conclusions}
In this article we proved that the known relation between BRST cohomology and Lie algebra cohomology for irreducible theories can be recovered also in the case of reducible theories. This goal has been achieved via the introduction of a generalize notion of Lie algebra cohomology, which allows to have generators of higher degree. Even though we restricted ourselves to consider a {\small{$U(2)$}}-matrix model, we expect that the same approach could be followed for other classes of models, with higher level of reducibility. By separating the BRST generators by their ghost degree, we believe that the induced gauge-fixed BRST complex can then be seen as a weighted multicomplex for this generalized Lie algebra complex, where the weights of the indices are determined by the ghost degree of the generators while the parity of the complex reflects the parity, bosonic or fermionic, of the BRST generators. The hope and the believe is that, thanks to this approach, we can have a clearer understanding of this gauge-fixed BRST complex and of its structure also for higher order gauge-symmetry groups. Moreover, as shown for our model of interest, this point of view could be helpful in facing the challenging task of explicitly describing the BRST cohomology groups of degree {\small{$k>0$}} and eventually contributing to clarifying their physical relevance, which has still be to fully understood.

\bibliographystyle{plain}

\end{document}